\documentclass[journal]{IEEEtran}
\IEEEoverridecommandlockouts
\usepackage{cite}
\usepackage{amsmath,amssymb,amsfonts,bbm}
\usepackage{algorithmic}
\usepackage{graphicx}
\usepackage{subfigure}
\usepackage{parskip}
\usepackage{textcomp}
\usepackage{xcolor}
\def\BibTeX{{\rm B\kern-.05em{\sc i\kern-.025em b}\kern-.08em
    T\kern-.1667em\lower.7ex\hbox{E}\kern-.125emX}}
\usepackage{amsthm}
\usepackage[ruled,linesnumbered]{algorithm2e}
\theoremstyle{definition}
\newtheorem{thm}{Theorem}
\newtheorem{rem}{Remark}
\newtheorem{example}{Example}
\newtheorem{definition}{Definition}
\begin{document}
\title{Local to Global: A Distributed Quantum Approximate Optimization Algorithm for Pseudo-Boolean Optimization Problems
%\thanks{This work is supported by the National Natural Science Foundation of China (NSFC) under Grants No. 62273226 and No. 61873162.}
}
\author{\IEEEauthorblockN{Bo Yue,
Shibei Xue~\IEEEmembership{Senior Member,~IEEE}, Yu Pan~\IEEEmembership{Senior Member,~IEEE}, Min Jiang, Daoyi Dong~\IEEEmembership{Fellow,~IEEE}}
\thanks{This work was supported by the National Science Foundation of China under Grant 62273226, Grant 62173296 and Grant 61873162. ({\it Corresponding author: Shibei Xue.})}
\thanks{Bo Yue and Shibei Xue are with the Department of Automation, Shanghai Jiao Tong University and the Key Laboratory of System Control and Information Processing, Ministry of Education of China, Shanghai, 200240, P. R. China (E-mails: \{yuebo2017JD, shbxue\}@sjtu.edu.cn).}
\thanks{Yu Pan is with Institute of Cyber-Systems and Control,
College of Control Science and Engineering, Zhejiang University, Hangzhou 310027, P. R. China (E-mail: ypan@zju.edu.cn).}
\thanks{Min Jiang is with School of Electronics and Information Engineering, Soochow University, Suzhou 215006, P. R. China (E-mail: jiangmin08@suda.edu.cn).}
\thanks{Daoyi Dong is with the School of Engineering, Australian National University, Canberra ACT 2601, Australia. (E-mail: daoyi.dong@anu.edu.au).}}

\maketitle

\begin{abstract}
% With the growing maturity of quantum computing, Quantum Approximate Optimization Algorithm (QAOA) is proposed as a promising candidate to demonstrate quantum supremacy, which tackles combinatorial optimization problems by encoding the approximate solution as a quantum ansatz state to be variationally optimized. Nonetheless, currently available Near-term Intermediate Scale Quantum(NISQ) devices confront two severe challenges, namely the limited number of qubits and restricted coherence time.
With the rapid advancement of quantum computing, Quantum Approximate Optimization Algorithm (QAOA) is considered as a promising candidate to demonstrate quantum supremacy, which exponentially solves a class of Quadratic Unconstrained Binary Optimization (QUBO) problems.
%, such as scheduling, resource allocation, etc.
However, limited qubit availability and restricted coherence time challenge QAOA to solve large-scale pseudo-Boolean problems on currently available Near-term Intermediate Scale Quantum (NISQ) devices. In this paper, we propose a distributed QAOA which can solve a general pseudo-Boolean problem by converting it to a simplified Ising model. Different from existing distributed QAOAs' assuming that local solutions are part of a global one, which is not often the case, we introduce community detection using Louvian algorithm to partition the graph where subgraphs are further compressed by community representation and merged into a higher level subgraph. Recursively and backwards, local solutions of lower level subgraphs are updated by heuristics from solutions of higher level subgraphs. Compared with existing methods, our algorithm
incorporates global heuristics into local solutions such that our algorithm is proven to achieve a higher approximation ratio and outperforms across different graph configurations. Also, ablation studies validate the effectiveness of each component in our method.
\end{abstract}

\begin{IEEEkeywords}
Quantum approximate optimization algorithm, distributed algorithm, pseudo-Boolean optimization.
\end{IEEEkeywords}

\section{Introduction}

% QC-QC取得的成就-NISQ（error）-QAOA相对适用当前时代-但是对于large-scale 还是无能为力-现阶段处理large-scale的方法综述- 为什么要提出distributed QAOA-有什么好处- 我们做了什么-未来愿景
Quantum Computing (QC) has emerged as a new computational paradigm that holds the potential of exponential speedups over classical counterparts in certain computationally-intensive tasks. This is primarily attributed to the vast state space of qubits and the inherent quantum parallelism.
%, since basically each additional qubit enlarges the size of computational state space in a quantum algorithm.
Earliest achievements include Shor's algorithm for exponential speedup of large integer factorization~\cite{b0,b1} and  Grover's algorithm for quadratic acceleration of unstructured database searching~\cite{b2}. Recent works include reduced computational time for chemistry~\cite{b3,b4,b5}, machine learning~\cite{b-6,10081,10083}, finance\cite{b-10,b-11} and other fields\cite{b-8,10082}. However, these quantum algorithms rely on scalable, fault-tolerant, universal quantum computers, which are not available at present.

Current quantum devices are not yet advanced enough for fault-tolerance and typically consist of a moderate number of noisy qubits. Consequently, this current stage of QC is referred to as Noisy Intermediate-Scale Quantum (NISQ)~\cite{b---1} era.
In this era, Quantum Approximate Optimization Algorithm (QAOA)~\cite{b11} becomes a leading candidate to demonstrate quantum supremacy. This expectation is grounded in the fact that QAOA requires relatively shallow circuits, making it suitable for non-error corrected devices. Additionally, the variational nature of QAOA helps mitigate the impact of systematic errors~\cite{b-12,10084}. However, the state-of-the-art quantum computers can only provide a limited number of qubits, which is approximately one hundred~\cite{b-15}. Hence, limited qubit count and restricted coherence time pose challenges for QAOA to outperform classical solvers, when confronted with large-scale pseudo-Boolean optimization problems~\cite{b-13,b-14}.
%which tackles NP-hard pseudo-Boolean problems, proven to be ,

Given the current circumstances, distributed QAOA becomes a viable alternative, which can be classified into circuit-level and graph-level algorithms. For the former kind of algorithm, CutQC cuts large quantum circuits into small subcircuits to be executed on smaller quantum devices, and utilizes classical postprocessing to reconstruct the output of the original circuit~\cite{b-16}. Similarly, Ref.~\cite{b-17} proposed a cluster simulation scheme that decomposes the overall circuit into clusters of bounded size and limited inter-cluster interactions. Because QAOA involves an unknown depth of iterations and pseudo-Boolean optimization problems have distinctive mathematical expressions, graph-level algorithms, based on the divide-and-conquer principle, are preferred. Ref.~\cite{b-18} partitions the original graph into several subgraphs that share common nodes, solves the subproblems on subgraphs with QAOA and then combines local solutions. However, the sample complexity grows with the number of common nodes, which is far from feasible for large-scale instances with dense graphs. To solve complicated graph instances, Ref.~\cite{b-15} proposed QAOA-in-QAOA, which views local solutions as a node, and casts the merging process into a new Max-Cut problem to be addressed by QAOA.
Nonetheless, existing methods mainly focus on Max-Cut problems and usually is not applicable to more diverse and complicated problems, e.g., pseudo-Boolean problems. Furthermore, the performance of distributed QAOA may degrade, because current methods fix local solutions at the interconnections between subgraphs. This fixation, or naive combination, can lead to deviation from the global optimum, especially when interconnections involve weighted edges or the subgraphs are not sufficiently dense..

%Nonetheless, the method fixes local solutions when consider interconnections between subgraphs, which deviates from global optimum when interconnections are weighted edges or the subgraphs are not dense enough. Additionally, ~\cite{b-15} is confined to MaxCut problems only.

In this paper, we extend the application of QAOA to address general pseudo-Boolean problems by transforming them into Quadratic Unconstrained Binary Optimization (QUBO) problems. By partitioning the original graph for the obtained QUBO problem, we propose a distributed QAOA employing Louvain algorithm as a community detection technique, which ensures dense structures within subgraphs. Additionally, we introduce community representation, which compresses lower level subgraphs and then merges them into higher level subgraphs. In this process, we refine the merging process to accommodate the limited qubit count. After recursively implementing this process, local solutions of lower level subgraphs are updated using solutions from higher level subgraphs to obtain the final global solution in a backward manner. Since the workflow of distributed QAOA is modularized, our method is also at ease of integration under future advancements.

The paper is organized as follows. Section~\ref{pre} briefly introduces the formulation of pseudo-Boolean problems. Section~\ref{2} explains how to convert a general pseudo-Boolean problem to a simplified Ising model. The distributed QAOA is presented in detail in Section~\ref{3}. Simulations and analysis are discussed in Section~\ref{4}. In the same section, an example on cubic knapsack problem shows the entire procedure and effectiveness of our method. Finally, Section~\ref{5} concludes the paper and suggests future research directions.

\section{Pseudo-Boolean Problem}\label{pre}
%A pseudo-Boolean problem aims to find the minimum (or maximum) of a pseudo-Boolean function while satisfying certain constraints of variables involved in the function, which are also expressed as pseudo-Boolean function.
A pseudo-Boolean problem refers to an optimization problem where the objective function and constraints are expressed using pseudo-Boolean functions.
Pseudo-Boolean problems have various applications in different fields such as computer science, operations research, artificial intelligence, and economics, and are used to model a wide range of real-world problems, including constraint satisfaction~\cite{b01}, combinatorial optimization~\cite{b02}, scheduling~\cite{b03}, and resource allocation~\cite{b04,b06}.
%a discrete but possibly large feasible domain.
%The objective function of a wealth of CO problems can be compactly and conveniently represented by pseudo-Boolean functions~\cite{b8}.
%\textcolor{red}{

Mathematically, in a pseudo-Boolean problem, we seek to find the minimum or maximum of a pseudo-Boolean function~\cite{b9}
\begin{equation}\label{multilinear}
    {f}_0(x_1,x_2,\ldots,x_K) = \sum_{{\mathcal{S}_0^l}\in 2^{[{K}]}}{c_l}\prod_{i\in {\mathcal{S}_0^l}}x_i.
\end{equation}
Here, ${f}_0:\mathbb{B}^K\rightarrow\mathbb{R}$ is a pseudo-Boolean function in a multilinear polynomial form, where $\mathbb{B}=\{0,1\}$ is a Boolean domain, and ${K}\in\mathbb{Z}^+ $ is called the arity of the function.
%A pseudo-Boolean function can be written as~(\ref{multilinear})
Also, $[{K}]$ is a set of natural numbers from $1$ to ${K}$; i.e., $[{K}] = \{1,2,\ldots,{K}\}$ and the power set $2^{[{K}]}$ includes all suitable subsets of $[{K}]$; i.e., $2^{[{K}]} = \{{\mathcal{S}_0^l}\,|\,{\mathcal{S}_0^l}\subseteq[{K}],\, l\in[L]\}$, where $l$ labels the subsets and $L$ indicates the total number of the subsets. In addition, a Boolean variable $x_i$ takes the value of either $0$ or $1$; i.e., $x_i\in\mathbb{B}$, where the number of the index $i$ is determined by $S_0^l$, the size of ${\mathcal{S}_0^l}$, and ${c_l}\in \mathbb{R}$ are the corresponding coefficient.
Note that since maximizing a function is equivalent to minimizing the negative counterpart of the function, a pseudo-Boolean problem can be considered as the problem of minimizing the function ${f}_0$.

Besides the objective function, a feasible domain should be considered in a pseudo-Boolean problem, which is shaped by both equality and inequality constraints. Generally, an inequality constraint can be expressed as
${g}_0(x_1,x_2,\ldots,x_K)\leq0$, where ${g}_0:\mathbb{B}^K\rightarrow\mathbb{R}$ is also a pseudo-Boolean function. We can always convert the inequality constraint to an equality constraint $g(x_1,x_2,\ldots,x_K,x_{K+1},\cdots,x_N)=0$ by introducing $N-K$ slack variables; i.e., $x_{K+1},\cdots,x_N$.
%Based on the two facts that i) $f_i(x_1,x_2,\ldots,x_n)\geq0$ can be transformed into $-f_i(x_1,x_2,\ldots,x_n)\leq0$, all inequality constraints are formulated in the form of $f_i(x_1,x_2,\ldots,x_n)\leq0$ for concision;
% ii) an inequality constraint $f_i(x_1,x_2,\ldots,x_n)\leq0$ can be converted to an equality constraint $f_{i}(x_1,\ldots,x_k,x_{k+1},\ldots,x_n)=0$ where $x_{k+1},\ldots,x_n$ are slack variables.
Consequently, a general constrained pseudo-Boolean problem is formulated as,
\begin{eqnarray}\label{cons}
  \min_{x_1,\cdots,x_N} &&{f}_0(x_1,x_2,\ldots,x_K)\nonumber\\
  s.t.\,\,\,\,\,&&g_1(x_1,x_2,\ldots,x_N)=0,\nonumber\\
  &&~~~~~~~~~\vdots\\
  &&g_W(x_1,x_2,\ldots,x_N)=0\nonumber
\end{eqnarray}
%\begin{align}\label{}
%\begin{array}{ll}
%    \mathop{\min}&\;
%    \text{s.t.}   &\\
%                  &\qquad \qquad \cdots\\
%                  &f_W(x_1,x_2,\ldots,x_n)=0,
%\end{array}
%\end{align}
where $W$ is the total number of constraints, and $g_i(x_1,x_2,\ldots, x_N),i\in[W]$ is also a multilinear polynomial.

The above constrained problem can be converted to an unconstrained problem by a penalty method which adds the constrains to the objective function as penalty terms~\cite{b07}; i.e.,
%To facilitate the resolution process, the penalty method is utilized to transform a constrained problem into an unconstrained problem by adding penalty terms that penalize violations of the constraints. Specifically, Eq.~(\ref{cons}) is converted equivalently into
\begin{equation}\label{gpb}
\begin{aligned}
\min f=\,\,&{f}_0(x_1,x_2,\ldots,x_K)+\mu\sum_{w=1}^Wg_w^2(x_1,x_2,\ldots,x_N)%\\
%=\,&\sum_{S_l\in 2^{[N]}}c_l\prod_{i\in S_l}x_i
\end{aligned}
\end{equation}
with a sufficiently large positive number $\mu\in\mathbb{R}^+$.
 %$C_i, \,i\in[W]$ are large positive numbers, and the domain is $\{x_1,x_2,\ldots,x_n\}$, which is omitted for simplicity.
Since both ${f}_0$ and $g_w, w\in[W]$ are multilinear polynomials and any power of a Boolean variable equals to itself, $f$ is also a multilinear polynomial, which can also expressed as
\begin{equation}\label{fff}
\begin{aligned}
f(x_1,x_2,\cdots,x_N)=\,&\sum_{\mathcal{S}^l\in 2^{[N]}}d_l\prod_{i\in \mathcal{S}^l}x_i.
\end{aligned}
\end{equation}
Here, due to the slack variables, $\mathcal{S}^l$ is a suitable subset of $[N]$ with a size $S^l$ and we use $d_l$ to denote the corresponding coefficients.
\section{Reduction of a Pseudo-Boolean Problem into a Simplified Ising Model}\label{2}
In the general pseudo-Boolean problem, the objective $f$ in Eq.~(\ref{fff}) would contain high-order terms of Boolean variables. However, due to hardware constraints in NISQ devices; i.e., high-order couplings are quite weak, most problems to be solved are considered in the form of Ising model~\cite{b7}. In this section, we propose a pipeline of reducing a general pseudo-Boolean problem into a simplified Ising model. %, which not only simplifies the original problem but also renders it solvable by distributed QAOA.
Concretely, Subsection \ref{pipe1} simplifies the objective function $f$ by eliminating uncoupled variables. Subsection~\ref{pipe2} reduces the simplified problem into an Ising model via quadratization and mapping from a QUBO problem to an Ising model. Subsection~\ref{pipe4} simplifies the Ising model in terms of variables with one dependency.
\subsection{Elimination of Uncoupled Variables}\label{pipe1}
%Since we have introduced slack variables in Section \ref{pre}
In the objective $f$, some variables take values independent of any other variables, so the values of these variables can be determined as $0$ or $1$ in advance so as to simplify the objective $f$. These variables refer to uncoupled variables.
%The mathematical definition of uncoupled variable and a theorem regarding the equivalence of this replacement are stated below.
\begin{definition}[Uncoupled Variable] In a multilinear polynomial $f(x_1,\ldots,x_N)$, a Boolean variable $x_i,i\in[N]$ is said to be an uncoupled variable if and only if
\begin{equation}
    f(1,\ldots,x_i=1,\ldots,1)-f(1,\ldots,x_i=0,\ldots,1)=1,\label{def1}
\end{equation}
when we let $d_l=1,\forall\, l\in [L]$.
\end{definition}
%\begin{definition}[Uncoupled Variable] In a multilinear polynomial $f(x_1,\ldots,x_N)$, $\forall\, l\in [L]$, supposing $d^l=1$, we have
%\begin{equation}
%    f(1,\ldots,x_i=1,\ldots,1)-f(1,\ldots,x_i=0,\ldots,1)=1,
%\end{equation}
%and thus a Boolean variable $x_i,i\in[N]$ is said to be an uncoupled variable.
%\end{definition}

\begin{thm}
Let $x_i$ be an uncoupled variable in $f(x_1,\ldots,x_N)$, and the coefficient of the term it belongs to be $d_{l}$, we have $\min f(x_1,\ldots,x_N)=\min f(x_1,\ldots,x_i=-\frac{d_{l}}{2|d_{l}|}+\frac{1}{2},\ldots,x_N)$.
\end{thm}
\begin{proof}
First, Eq. (\ref{def1}) in the definition indicates that $x_i$ exists only in one term of $f$. This is because $f(1,\ldots,x_i=1,\ldots,1)$ equals to the total number of terms in $f(x_1,\ldots,x_N)$, while $f(1,\ldots,x_i=0,\ldots,1)$ equals to the total number of terms that exclude $x_i$. %Hence, Eq.(4) indicates $x_j$ exists only in one term in $f(x_1,\ldots,x_N)$.
%By definition,.
Further, the value of the term containing $x_i$; i.e., $d_lx_i\prod_{j\in \mathcal{S}^l \backslash \{i\}}x_j$ is either $0$ or $d_l$ since $x_i\in\mathbb{B},i\in[N]$, where $\mathcal{A}\backslash \mathcal{B}=\{x|x\in \mathcal{A}, x\notin \mathcal{B}\}$ for two sets $\mathcal{A}$ and $\mathcal{B}$. When $d_l>0$, the term with $x_i=0$ is no greater than the term with $x_i=1$, namely $0\leq d_l\prod_{j\in \mathcal{S}^l\backslash \{i\}}x_j$.
Hence, when minimizing $f$, we have %and $x_i$ exists only in this term, thus
$x_i=0$ when $d_l>0$. Similarly, if $d_l<0$, the term with $x_i=1$ is no greater than the term with $x_i=0$, namely $d_l\prod_{j\in \mathcal{S}^l}x_j\leq0$. To minimize $f$, we have $x_i=1$ when $d_l<0$. Therefore, we choose $x_i=0$ when $d_l>0$, and $x_i=1$ when $d_l<0$, which can be written in a compact form as  $x_i=-\frac{d_l}{2|d_l|}+\frac{1}{2}$.
\end{proof}
This theorem shows that the value of an uncoupled variable can be determined in advance when we minimize the objective $f$.
We introduce an operation $\xi$ which lets all uncoupled variables equal to $-\frac{d_l}{2|d_l|}+\frac{1}{2}$. Hence, the original problem can be simplified as
%Denote a function that replaces, we derive
\begin{equation}
\begin{aligned}
    &\min \;\;f(x_1,x_2,\ldots,x_N)\\
    \Longleftrightarrow&\min \;\;\xi\circ f(x_1,x_2,\ldots,x_N)\\
    \Longleftrightarrow&\min\;\;\hat{f}(x_{1},x_{2},\ldots,x_{P}),\;P \leq N
\end{aligned}
\end{equation}
where $\hat{f}$ is also a multilinear polynomial and only contains coupled variables. Assuming there are $N-P$ uncoupled variables, this replacement reduces the number of variables from $N$ to $P$, which predigests the objective function $f$. Note that the operation $\xi$ does not introduce new uncoupled variables since the remaining variables do not satisfy Eq.~(\ref{def1}).
%Consequently, the objective function $f$ is simplified as follows,
%\begin{align}
%\begin{split}
%\begin{array}{l}
%\qquad \quad\;\;\,\mathop{\min}\;\;\;
%    f(x_1,x_2,\ldots,x_N)\\
%\\\iff
%\begin{array}{ll}
%    \;\,\,\mathop{\min}
%    &f(x_1,x_2,\ldots,x_N)\\
%    \;\,\,\text{s.t.} &x_{j_1}=-\frac{c_{l_1}}{2|c_{l_1}|}+\frac{1}{2},\;l_1\in[L]\\
%    &\qquad\quad\cdots\\
%    &x_{j_{N-P}}=-\frac{c_{l_{N-P}}}{2|c_{l_{N-P}}|}+\frac{1}{2},\\
%    \;&l_{N-P}\in[L]
%\end{array}
%\\\iff
%    \quad\,\mathop{\min}\;\;\;
%    \hat{f}(x_{r_1},x_{r_2},\ldots,x_{r_{P}}), P\leq N.\\
%\end{array}
%\end{split}
%\end{align}
%Note that $\hat{f}$ .
\begin{example} Let us consider $f(x_1,x_2,x_3,x_4)=x_1x_4-2x_2x_3+4x_1x_2x_4$. By definition, $x_3$ is an uncoupled variable, so $\min\, \xi\circ f(x_1,x_2,x_3,x_4) = \min\, \hat{f}(x_1,x_2,x_4)$, where we let $x_3=1$ and thus $\hat{f}(x_1,x_2,x_4)=x_1x_4-2x_2+4x_1x_2x_4$.
\end{example}

\subsection{Quadratization}\label{pipe2}
% Since QAOA adopts a Hamiltonian(see Eq.~(\ref{eq:Ham})) typically generated from Ising model(see Eq.~(\ref{eq:Isi})) ~\cite{b7} which can be directly converted from Quadratic Unconstrained Binary Optimization(QUBO) model~\cite{b6},  $u(x_1,x_2,...,x_m),x_\cdot\in\{0,1\}$ via $x_\cdot=\frac{(1-z_\cdot)}{2}$, we need to ensure the maximum degree of each monomial in $\hat{h}$ to be two, which can be achieved by means of quadratization.
% \begin{equation}
%     g(z_1,z_2,\cdots,z_m)=-\sum_{1\leq i<j\leq m}\alpha_{ij}z_iz_j-\sum_{k=1}^{m}\beta_iz_k,z_{\cdot}\in\{-1,1\},
%     \label{eq:Isi}
% \end{equation}
% \begin{equation}
%     H=g(\sigma_1^z,\sigma_2^z,\cdots,\sigma_m^z),\;\sigma^{z}_{\cdot}=\begin{pmatrix} 1 & 0 \\ 0 & -1 \end{pmatrix}.
%     \label{eq:Ham}
% \end{equation}
%\textcolor{red}{Please show the general procedure. it seems the following description is for cubic terms. How to quadratize higher-order terms. Some expressions would not be correct as highlighted by red fonts.}

As aforementioned, the objective function Eq.~(\ref{fff}) would contain high-order terms, e.g., cubic terms. In the current state, to execute a QAOA on NISQ devices, it is expected to quadratize these high-order terms. %Rosenberg proposes a procedure, which can efficiently reduce a multilinear polynomial into a quadratic polynomial
The following quadratization procedure is proposed in Ref.~\cite{b10}, where three steps are performed for one iteration.
%, wh takes quadratization of one high-order term as an example and can be applied to others.
Concretely, we consider to quadratize a monomial $\prod_{i\in \mathcal{S}^l}x_i$, where the size of $\mathcal{S}^l$ is greater or equal to $3$; i.e., $S^l\geq3$.
%the arity of a multilinear polynomial denotes the total number of Boolean variables. Quadratization reduces \textcolor{red}{the arity of $\hat{f}$from $N$ to two}. A quadratized multilinear polynomial consists of monomials whose arity are at most two.
In the first step, we select two variables $x_i$ and $x_j,~i\neq j$ in the monomial. % whose arity is at least three in $\hat{f}$,
In the second step, we replace each occurence of $x_ix_j$ in $\hat{f}$ with an auxiliary Boolean variable $y_{ij}\in \mathbb{B}$; i.e., the monomial is written as $y_{ij}\prod_{p\in \mathcal{S}^l\backslash \{i,j\}}x_p$,  and denote the corresponding objective function as $\hat{f}_{ij}$. In the third step, we obtain a new objective $\tilde{f}=\hat{f}_{ij} +\lambda_{ij}(x_ix_j-2x_iy_{ij}-2x_jy_{ij}+3y_{ij})$ with a large positive number $\lambda_{ij}\in\mathbb{R}^+$, which is equivalent to $\hat{f}$ in the sense that $\min \hat{f} = \min \tilde{f}$. It is easy to check this equivalence. Since both $x_i$ and $x_j$ are Boolean variables, we have four combinations of $x_i$ and $x_j$, namely $00$, $01$, $10$, or $11$. Taking them to $\tilde {f}$, we can find the equivalence. %the minimum of $\hat{f}'$ over the auxiliary variable $y_{ij}$ is exactly equal to that of $\hat{f}$. %Indeed, the minimizer is $y^*_{ij}=x_ix_j$, and the penalty term vanishes for this value. To verify, suppose $x_i=1$ and $x_j=1$ for the optimal solution, the penalty term equals to $\lambda_{ij}(1-y_{ij})$, where, to minimize the penalty term, $y_{ij}=1=x_ix_j$. This holds still for $00$, $01$ and $10$.
The above three steps can be
repeated until we have no high-order terms. We denote this quadratized polynomial as $\tilde{f}(x_{1},\ldots,x_{{P}},{y}_1,\ldots,{y}_{U})$, where ${y}_i,i\in[U]$ are reordered auxiliary variables. Note that the auxiliary variables are also Boolean variables, and $\tilde{f}$ is also a multilinear polynomial.
\begin{example}
Consider the quadratization of $\hat{f}=x_1x_2x_3$. i) The product $x_1x_2$ is selected. ii) We replace $x_1x_2$ with an auxiliary variable $y_{12}$ and have $\hat{f}_{12}=y_{12}x_3$. iii) We rewrite the original problem as $\min\, \hat{f} = \min\hat{f}_{12} +\lambda_{12}(x_1x_2-2x_1y_{12}-2x_3y_{12}+3y_{12})$. Consequently, $\tilde{f}(x_1,x_2,x_3,{y_1})={y_1}x_3+\lambda_{12}(x_1x_2-2x_1{y_1}-2x_2{y_1}+3{y_1})$.
\end{example}
With the above procedure, we can convert the original problem to a QUBO model~\cite{b6} which is composed of a quadratic objective function of Boolean variables.
%\subsection{Max-Cut Problems}\label{pipe3}
To solve the QUBO model on a quantum device, we should map a QUBO model into an Ising model by changing variables; i.e., $x_\cdot(y_\cdot)=\frac{1}{2}(1-z_\cdot)$.
%for all variables in the QUBO model.
Since $\tilde{f}(x_{1},\ldots,x_{{P}},{y}_1,\ldots,{y}_{U})$ is quadratized and consists of only Boolean variables, $\tilde{f}$ can be converted into an Ising model $
v(z_1,\ldots,z_{P+U})$, taking the form of
% Consider an undirected and weighted graph $G(V,E)$, where $V$ is the vertex set (node set) and $E$ is the edge set, a cut is a partition of $V$ into two subsets $V_1$ and $V_2$, satisfying $V_1\bigcup V_2=V$ and $V_1\bigcap V_2=\varnothing$. A MaxCut problem finds a cut that maximizes the aggregated weights of edges that connect one vertex from $V_1$ and the other from $V_2$. The mathematical formulation of a MaxCut problem reads
% \begin{equation}
%     \max\, \frac{1}{2}\sum_{(i,j)\in E}w_{ij}(1-z_iz_j)=\max\,( c-\frac{1}{2}\sum_{(i,j)\in E}w_{ij}z_iz_j)
% \end{equation}
% %The Ising model ~\cite{b7} of an $N$-node MaxCut problem reads,
% % \begin{equation}
% %     g(z_1,\cdots,z_N)=\frac{1}{2}\sum_{(i,j)\in E}w_{ij}(1-z_iz_j)=c-\frac{1}{2}\sum_{(i,j)\in E}w_{ij}z_iz_j,
% %     \label{isingeq}
% % \end{equation}
% where $\omega_{ij}$ is the weight of edge $(i,j)$, $z_{i,j} \in\{-1,+1\}$ denotes which subset the vertex is in, and $c=\frac{1}{2}\sum_{(i,j)\in E}w_{ij}$. Here, we assume $i<j$ holds to ensure each edge is counted once. Since $c$ and $-\frac{1}{2}$ are constants, the optimization problem can be redefined as the minimization of a new objective function, namely $\min\, \sum_{(i,j)\in E}w_{ij}z_iz_j$. To be solved by QAOA, we need to formulate the Ising model ~\cite{b7} of an $N$-node MaxCut problem, which is
% \begin{equation}
%     v_0(z_1,z_2,\ldots,z_n)=\sum_{(i,j)\in E}w_{ij}z_iz_j.
%     \label{isingeq}
% \end{equation}

% Note that, besides quadratic terms, $v$ contains single variable terms, so we need to add these terms to Eq.~(\ref{isingeq}),
\begin{equation}
    v(z_1,z_2,\ldots,z_{P+U})=\sum_{1\leq i<j\leq P+U}w_{ij}z_iz_j + \sum_{1\leq k \leq P+U}w_kz_k,
    \label{isingeqquasi}
\end{equation}
where we rewrite the corresponding coefficients as $w_\cdot$. Note that with the mapping the domain of the Boolean variables is changed to $\mathbb{\bar B}=\{-1,1\}$; i.e., $z_\cdot\in\mathbb{\bar B}$. Also, since the mapping is linear, Eq.~(\ref{isingeqquasi}) is equivalent to the objective $\tilde f$.
% \begin{example}
% Suppose $\tilde{f}(x_1,x_2,x_3) = 3-4(x_1+x_2+x_3)+4(x_1x_2+x_1x_3+x_2x_3)$. Replace each $x_{\cdot}$ with $\frac{(1-z_{\cdot})}{2}$, we have $v(z_1,z_2,z_3)=z_1z_2+z_1z_3+z_2z_3$, which can be seen as a 3-Node MaxCut problem with three edges $(1,2)$, $(2,3)$ and $(1,3)$.
% \end{example}
\subsection{Delayed Decision of Variables with One Dependency}\label{pipe4}
In Eq.~(\ref{isingeqquasi}), some variables only exist in one quadratic term and thus we call these variables with one dependency. The values of these variables can be decided when the values of other variables have been determined for minimizing the objective. Concretely, we can consider the one dependency terms as constraints with which we minimize the remaining part of the objective; i.e., the original cost without one dependency terms.
By doing so, we may have some new one dependency variables such that the corresponding terms should also be considered as constraints. The above procedure can be repeated until the objective involves no one dependency terms.
The notion of variable with one dependency and a theorem regarding delayed decision of variables with one dependency are stated below.

\begin{definition}[Variable with One Dependency]
A variable $z_i\in\mathbb{\bar B}$ in Eq.~(\ref{isingeqquasi}) is said to be a variable with one dependency if $z_i$ only exists in one quadratic term of Eq.~(\ref{isingeqquasi}) or $z_i$ exists in two quadratic terms in Eq.~(\ref{isingeqquasi}) which include one quadratic term with the other variable with one dependency.
\end{definition}

\begin{thm}
Suppose we derive a simplified $v$, denoted as $\check{v}$, with delayed decision of all variables with one dependency, the formula holds,
\begin{align}
\begin{split}
\begin{array}{l}
\qquad \quad\;\;\,\mathop{\min}\;\;\;
    v(z_1,z_2,\ldots,z_{P+U})\\
\\\iff
\begin{array}{ll}
    \;\,\,\mathop{\min} &v(z_1,z_2,\ldots,z_{P+U})  \\
    \;\,\,\textup{s.t.}
    &z_{j_1}=\mathop{\arg\min} \eta_{i_1j_1}z_{i_1}z_{j_1},\\
    &~~~~~~~~~~~~i_1\in [P+U]\backslash\{j_1\}+\varnothing,\\
    &\qquad\quad\vdots\\
    &z_{j_{B}}=\mathop{\arg\min} \eta_{i_{B}j_{B}}z_{i_{B}}z_{j_{B}},\\
    &~~~~~~~~~~~~i_{B}\in[P+U]\backslash\{j_{B}\}+\varnothing,
\end{array} \\\iff
\quad\,\mathop{\min}\;\;\;
\check{v}(z_{q_1},z_{q_2},\ldots,z_{q_{D}}),D= P+U-B
\end{array}
\end{split}
\end{align}
where $z_{j_1},\ldots,z_{j_B}$ are variables with one dependency in $z_1,z_2,\cdots,z_{P+U}$, $\eta_{ij}$ is the coefficient of term $z_iz_j$, and $z_{q_1},z_{q_2},\ldots,z_{q_{D}}$ are the remaining variables. Here, $i_{\zeta}\in\varnothing$ for first-order terms in $v$, so $\eta_{i_{\zeta}j_{\zeta}}z_{i_{\zeta}}z_{j_{\zeta}}$ is equivalent to $\eta_{j_{\zeta}}z_{j_{\zeta}}$.
\end{thm}
\begin{proof} The proof of first equivalence is divided into two parts. First, we prove that $\min v_1+\min v_2\leq\min v$ holds. Second, we show that the equality can always be achieved. $i)$ Split $v$ into two parts, namely $v=v_1+v_2$. $v_1=\eta_{i_1j_1}z_{i_1}z_{j_1}+\ldots+\eta_{i_{B}j_{B}}z_{i_{B}}z_{j_{B}}$ contains all the terms that involve variables with one dependency, while $v_2=v-v_1$ contains the remaining terms. Define $Z_c \subseteq \{z_1,z_2,\ldots,z_{P+U}\}$ as the set of common variables whose elements appear in both $v_1$ and $v_2$.
Since $\forall z_c \in Z_c$ has to be assigned with the same value in $v$ but can be assigned with different values in $v_1$ and $v_2$ for respective minimizers, $\min v_1+\min v_2\leq\min v$ holds. %\textcolor{magenta}{This fact can also be proven from the perspective of counterargument. If $\min v_1+\min v_2>\min v$ holds, truncating the domain of $v$ to accommodate the domain of $v_1$ and $v_2$ will lead to $\min v_1+\min v_2=\min v$, which violates the assumption and thus verifies $\min v_1+\min v_2\leq\min v$.}
$ii)$ Suppose $v_2$ is minimized and $\forall z_c \in Z_c$ has been assigned with value. Then, by definition of variable with one dependency, for each term in $v_1$, the term can always take the minimum value based on the variable with one dependency no matter what values $z_c$ takes. This indicates that $v_1$ can always take the minimum value no matter what values $z_c$ take. Thus, common variables of respective minimizer of $v_1$ and $v_2$ can share the same value, which leads to $\min v_1+\min v_2=\min v$. This accomplishes the proof of the first equivalence. The second equivalence is straightforward, where we rewrite the objective by removing variables with one dependency.
\end{proof}
\begin{example}
Consider $v(z_1,z_2,z_3,z_4,z_5)=z_1z_2+z_2z_3+z_3z_1+z_3z_4+z_4z_5+z_2$. Since $z_5$ only exists in one quadratic term and $z_4$ exists in two terms one of which includes $z_5$, by definition of variable with one dependency, both $z_4$ and $z_5$ are variables with one dependency. Hence, $\min v(z_1,z_2,z_3,z_4,z_5)=\min \check{v}(z_1,z_2,z_3)$, with $z_4=\arg\min z_3z_4$ and $z_5=\arg\min z_4z_5$.
\end{example}
%In a MaxCut problem scenario, variables with one dependency appear as nodes in a chain.

The removal of variables with one dependency simplifies the Ising model. To be simple and vivid, we call delayed decision of variables with one dependency as elimination of chains afterwards. Ultimately,
%$\check{v}$ takes the form of Eq.~(\ref{isingeqquasi}) and can be regarded as a simplified Ising model of original pseudo-Boolean problem $f$.
%In the following discussion, we utilize the MaxCut problem to shed light on distributed QAOA, where the objective function does not have single variable terms, but this can be easily extended since single variable terms are always within the community and allowed in the Hamiltonian of QAOA.
$\check{v}$ can be expressed as,
\begin{equation}
    \check{v}(z_{q_1},z_{q_2},\ldots,z_{q_D})=\sum_{1\leq i<j\leq D}w_{ij}z_{q_i}z_{q_j} + \sum_{1\leq k \leq D}w_kz_{q_k}.
    \label{isingvcheck}
\end{equation}
To summarize, a general pseudo-Boolean problem $\min f$ in Eq.~(\ref{gpb}) is reduced into minimization of a simplified Ising model $\min \check{v}$ in Eq.~(\ref{isingvcheck}) via three procedures, i.e., elimination of uncoupled variables, quadratization and elimination of chains. We will find in the following section that these operations simplify QAOA and benefit distributed QAOA.

\section{Conventional QAOA}
In the above section, we have converted a common pseudo-Boolean problem into an Ising model which can be efficiently solved by QAOA. This is because QAOA combines efficient exploration of energy landscapes for non-convex optimization problems in quantum annealing~\cite{b--1} and error mitigation in variational quantum algorithms~\cite{b--2}. In this section, we will briefly review the conventional QAOA.

Since QAOA solves an optimization problem on a quantum system, it is necessary to map the objective of an optimization problem into the Hamiltonian of the quantum system. In conventional QAOA, Eq.~(\ref{isingvcheck}) is converted into a Hamiltonian
\begin{equation}
    H = \sum_{1 \leq i < j \leq D }\omega_{ij}\sigma_{q_i}^z\sigma_{q_j}^z+\sum_{1 \leq k \leq D}\omega_{k}\sigma_{q_k}^z,
    \label{eq:Hamiltonian}
\end{equation}
where the Hamiltonian is formulated by substituting Boolean variables $z_{q_i}z_{q_j}$ and $z_{q_k}$ with matrices $\sigma_{q_i}^z\sigma_{q_j}^z$ and  $\sigma_{q_k}^z$, respectively.
The Pauli-$z$ matrix for the $\theta$-th qubit $\sigma^z_{q_\theta}=I_1\otimes\ldots\otimes I_{\theta-1} \otimes \sigma^z \otimes I_{\theta+1}\otimes\ldots\otimes I_D$ with $\sigma^z = \begin{pmatrix} 1 & 0 \\ 0 & -1 \end{pmatrix}$ and identity matrix $I_j=\begin{pmatrix} 1 & 0 \\ 0 & 1 \end{pmatrix},\,j\in[D]\backslash\{\theta\}$. Here, $D$ denotes the total number of qubits.

A diagram of conventional QAOA is given in Fig.~\ref{fig:qaoa}, which sketches the procedure of QAOA. On a quantum device, traditional solutions to the original optimization problem are encoded as quantum states. In particular, the optimal solution $z^*\in \mathbb{\bar B}^D$ is encoded as the ground state of Hamiltonian.
%(or Ising model) $H = \sum_{(i,j)\in E}\omega_{ij}\sigma_{i}^z\sigma_{j}^z$.
QAOA approximates this ground state by an ansatz state~\cite{b11},
\begin{equation}
    |\psi(\vec{\gamma},\vec{\beta})\rangle=U_B(\beta_p)U_C(\gamma_p)\cdots U_B(\beta_1)U_C(\gamma_1)|\psi_0\rangle.
\end{equation}
The initial state is expressed as $|\psi_0\rangle=|+\rangle^D=|+\rangle_D\ldots|+\rangle_1$, where the superposition state is $|+\rangle=(|0\rangle+|1\rangle)/\sqrt{2}$ with the ground state $|0\rangle=(1,0)^T$ and the excited state $|1\rangle=(0,1)^T$ for each qubit. $p$ is the depth of the ansatz state.
Two kinds of operators are alternatively utilized by $p$ times to approach the ground state. One operator is the phase separating operator $U_B(\beta_j)=e^{-i\beta_j \sum_{\theta=1}^D\sigma^x_{q_\theta}},j\in[p]$ which alters the phase of potentially good quantum states, and the other is the mixing operator $U_C(\gamma_j)=e^{-i\gamma_{j}H},j\in[p]$ which increases the probability or weight of these good quantum states in a superposition state based on phases of quantum states. Here, the Pauli-$x$ matrix for the $\theta$-th qubit $\sigma^x_{q_\theta}=I_1\otimes\ldots\otimes I_{\theta-1} \otimes \sigma^x \otimes I_{\theta+1}\otimes\ldots\otimes I_D$ with matrix $\sigma^x = \begin{pmatrix} 0 & 1 \\ 1 & 0 \end{pmatrix}$. In the ansatz, two vectors of hyperparameters $\vec{\gamma} \in [0,\pi)^{\otimes p}$ and $\vec{\beta} \in [0,2\pi)^{\otimes p}$ are iteratively updated by a classical optimizer until a terminating condition is satisfied, such as the value of the objective function is less than a required threshold. The final ansatz state reads
%is a good approximation of the optimal solution
\begin{equation}
    |\psi(\vec{\gamma}^*,\vec{\beta}^*)\rangle=\arg\min\langle \psi(\vec{\gamma},\vec{\beta})|H|\psi(\vec{\gamma},\vec{\beta}) \rangle\approx|z^*\rangle,
\end{equation}
where $\vec{\gamma}^*$ and $\vec{\beta}^*$ are the optimized vectors of hyperparameters.
%and $\langle \psi(\vec{\gamma},\vec{\beta})|H|\psi(\vec{\gamma},\vec{\beta}) \rangle$ denotes the energy of the quantum system.
Then, the final ansatz state undergoes quantum measurement and collapses into an eigenstate of $H$, most possibly a good approximation of $H$'s ground state, after which the promising eigenstate is decoded into a classical solution which is an approximate minimizer of the objective function $\check{v}$ in Eq.~(\ref{isingvcheck}).

\begin{figure}[htbp]
    \centering
    \includegraphics[width=0.48\textwidth]{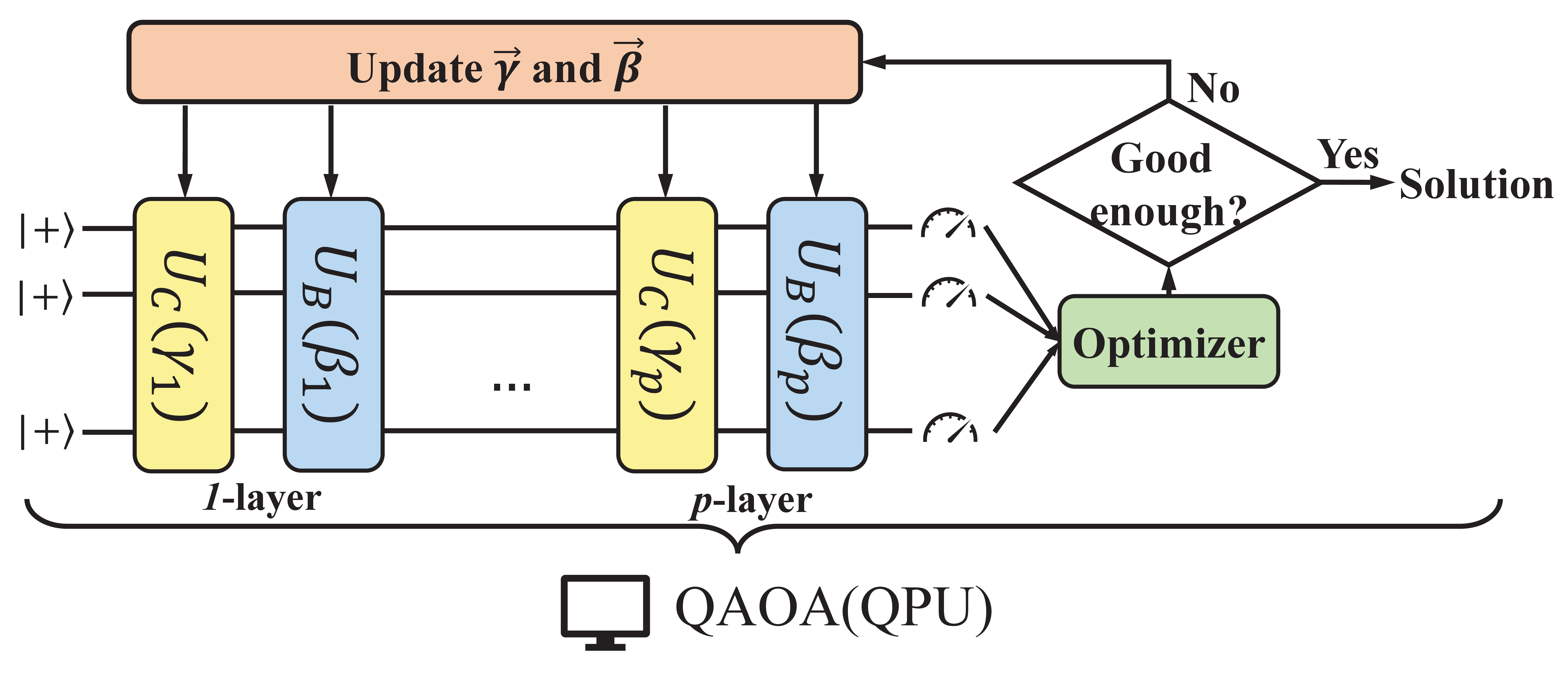}
    \caption{Sketch of conventional QAOA.}
    \label{fig:qaoa}
\end{figure}

% Take the Max-Cut problem as an example. Consider an undirected and weighted graph $G(V,E)$, where $V$ is the vertex set (node set) and $E$ is the edge set, a cut is a partition of $V$ into two subsets $V_1$ and $V_2$, satisfying $V_1\bigcup V_2=V$ and $V_1\bigcap V_2=\varnothing$. A Max-Cut problem finds a cut that maximizes the aggregated weights of edges that connect one vertex from $V_1$ and the other from $V_2$. The mathematical formulation of a Max-Cut problem reads
% \begin{equation}
%     \max\, \frac{1}{2}\sum_{(i,j)\in E}w_{ij}(1-z_iz_j)=\max\,( c-\frac{1}{2}\sum_{(i,j)\in E}w_{ij}z_iz_j)
% \end{equation}
% where $\omega_{ij}$ is the weight of edge $(i,j)$, $z_{i,j} \in\{-1,+1\}$ denotes which subset the vertex is in, and $c=\frac{1}{2}\sum_{(i,j)\in E}w_{ij}$. Here, we assume $i<j$ holds to ensure each edge is counted once. Since $c$ and $-\frac{1}{2}$ are constants, the optimization problem can be redefined as the minimization of a new objective function, namely $\min\, \sum_{(i,j)\in E}w_{ij}z_iz_j$. To be solved by QAOA, we need to formulate the Ising model  of an $N$-node Max-Cut problem, which is
% \begin{equation}
%     v(z_1,\cdots,z_N)=\sum_{(i,j)\in E}w_{ij}\sigma_i^z\sigma_j^z,
%     \label{isingeq}
% \end{equation}
% Fig.~\ref{fig:qaoa} demonstrates the pipeline of conventional QAOA and its application in a $9$-node Max-Cut problem.

\section{Distributed QAOA}\label{3}
Due to restricted qubit count and limited coherence time of NISQ devices, it is not feasible to solve large-scale pseudo-Boolean problems in a centralized quantum device. An alternative way is to explore distributed QAOA.
Existing methods partition an original graph into several subgraphs through random partition or greedy modularity maximization algorithm. The solution of each subgraph is obtained by conventional QAOA, after which each subgraph is viewed as a single vertex and the global solution is obtained by directly merging these local solutions, which is again a problem that can be solved by QAOA. However, the two graph partition techniques tend to produce sparse or large subgraphs that can degrade distributed QAOA. Besides, this brute combination of solutions for each subgraph despise connections (weighted edges) between subgraphs and there are conditions when the weights of inter edges outweigh that of inner edges, resulting in deviation from the global optimum. Finally, existing methods do not take into account the limited qubit count when merging local solutions.
%The existing method~\cite{b-15} is a pioneer in this field and has achieved better results than other competitors, but there are several drawbacks.
%In this section, we propose distributed QAOA, which consists of a superior community detection technique, an innovatively proposed community representation and update technique, and a refined combination process.
%First, the existing method only considers Max-Cut problems which have limited applications. Our method expands applicable scenarios to pseudo-Boolean problems, which is given in Section~\ref{2}.
%Second, the graph partition techniques in the existing method have disadvantages that can degrade the performance of distributed QAOA. Third, the existing method fixes local solutions which neglects heuristics from larger graphs, thus deviating form global optimum. Finally, the existing method fails to take into account qubit count when merging local solutions.}

%上一段介绍了研究动机，由此给出我们的分布式QAOA的流程。首先将原问题表示成图。然后通过L算法进行commnity detection，以获得更密集的子图。之后，在每个子图上用QAOA求解后，将相应解合并。最后zenmyang可以得到近似解。
With respect to the above points, we present a framework of distributed QAOA, where we tackle the above drawbacks by community detection and accommodation, community representation and update in Section~\ref{31}, and~\ref{32}, respectively. For the sake of clarity and intuition, we use a graph $G$ to represent the simplified Ising model in Eq.~(\ref{isingvcheck}) which contains both quadratic terms and first-order terms. In our setting, each variable is assigned with a node. Each quadratic term in Eq.~(\ref{isingvcheck}) are represented in the graph as two nodes and an edge connecting them, where the weight of the edge is the coefficient of the quadratic term. Hence, we have $G=(V,E)$ where $V$ is the set of all nodes and $E$ is the set of all edges in the graph. In addition, each first-order term in Eq.~(\ref{isingvcheck}) is bonded to the node in the graph according to the variable. Each node takes value from either $-1$ or $+1$, after which the value of in Eq.~(\ref{isingvcheck}) is obtained. With the graphic representation of the problem,
Fig.~\ref{fig:workflow} presents the basic working procedure of distributed QAOA, which generates a global solution to the original graph, $i)$ the community detection technique efficiently and effectively partitions the graph into subgraphs where QAOA is employed to obtain local solutions; $ii)$ community representation iteratively compresses and merges subgraphs at different levels. Subsequently, higher-level heuristics update the local solutions of these subgraphs until a global solution is achieved. This hierarchical approach enhances global optimization by providing global heuristics to local solutions, enabling a well-approximated global optimum.
%, and $iii)$ the refined combination process ensures qubit count is not exceeded when merging subgraphs.

% In this part, we summarize the workflow of distributed QAOA from a holistic perspective, which includes how to combine local solutions and update through layers of combination procedure, as elucidated in Fig.~\ref{fig:workflow}.

\begin{figure*}[ht]
    \centering
    \includegraphics[width=0.95\textwidth]{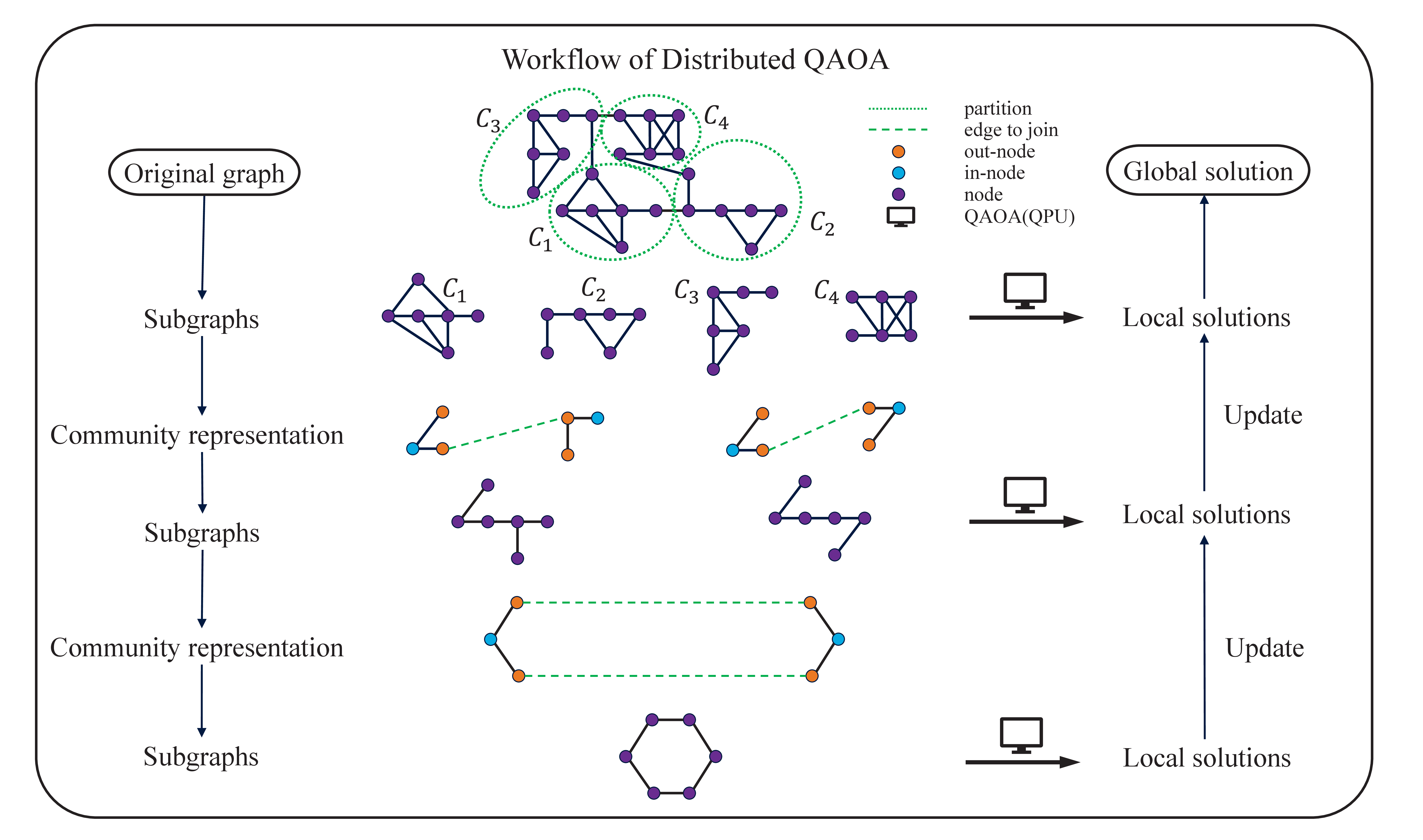}
    \caption{Workflow of distributed QAOA. The workflow is tree-structured. An example of workflow of distributed QAOA concerning $24$-node Max-Cut with allowed maximum number of $6$ qubits is given. The original graph is displayed in the upper middle. Community representation is done within each layer from bottom to top of the tree structure, after which update is done from top to bottom. The computer icon denotes the current graph undergoes QAOA on a local quantum computer. The legend is on the left of upper right corner.}
    \label{fig:workflow}
\end{figure*}

% \textcolor{red}{Fig.~\ref{fig:diagram} presents the basic working procedure of distributed QAOA, where the community detection technique partitions the graph into subgraphs, community representation compresses, merges and later updates the local solutions of subgraphs, and the refined combination process ensures qubit count is not exceeded when merging subgraphs. QAOA is performed to obtain an approximate solution on subgraphs.}

%For the sake of clarity and intuition, we use a graph setting to represent the simplified Ising model in Eq.~(\ref{isingvcheck}) which contains both quadratic terms and first-order terms. In our setting, quadratic terms in Eq.~(\ref{isingvcheck}) are represented in the graph in the exact way of graphs for Max-Cut problems. However, each first-order term is bonded to the node in the graph corresponding to the variable, which does not exist in the Max-Cut setting. Fig.~\ref{fig:diagram} presents the basic working procedure of distributed QAOA, where the community detection technique partitions the graph into subgraphs, community representation compresses, merges and later updates the subgraphs, and the refined combination process ensures qubit count is not exceeded when merging subgraphs. QAOA is performed to obtain an approximate solution on subgraphs.

% \begin{figure}[htbp]
%     \centering
%     \includegraphics[width=0.35\textwidth]{diagram2.eps}
%     \caption{Diagram of distributed QAOA.}
%     \label{fig:diagram}
% \end{figure}

\subsection{Community Detection and Accommodation}\label{31}
 Community detection aims to partition the original graph into several subgraphs so that every subgraph can be solved by qubit-limited quantum devices. Suppose the original graph $G$ is partitioned into $M$ subgraphs (or communities) $\{G_k\}^M_{k=1}$, satisfying $V = \bigcup_{k=1}^{M} C_k$ with $\bigcap_{k=1}^{M} C_k = \varnothing$ where $C_k$ is the vertex set of subgraph $G_k$. Define $c_i$ as the subgraph that vertex $i$ belongs, and the indicator function $\delta(c_i,c_j)= \left\{
\begin{aligned}
    1, \;&c_i= c_j\\
    0, \;&c_i\neq c_j
\end{aligned} \right.$.
For edge $(i,j)$, if $\delta(c_i,c_j)=1$, edge $(i,j)$ belongs to the same subgraph; otherwise, edge $(i,j)$ belongs to adjacent subgraphs.
%\textcolor{red}{Additionally, suppose available quantum devices enable a qubit count of $q$ at maximum.}\textcolor{blue}{this sentence looks too lonely. Please give the reason why you need this description.}

Existing methods subsume two partition techniques, which are random partition and greedy modularity maximization~\cite{b-15}. Random partition successively samples $q$ vertices as a subgraph until $\lceil N/q \rceil$ subgraphs are generated, where $q$ is the maximum allowed qubit count and $\lceil \cdot \rceil$ is the ceiling function. However, random partition degrades the performance when the graph is sparse and inner subgraphs contain
less edges when inter subgraphs contain
more edges. To measure the quality of partitioning~\cite{b12}, researchers define modularity as
\begin{equation}
    Q=\frac{1}{2m}\sum_{(i,j)}\left[\omega_{ij}-\frac{k_ik_j}{2m}\right]\delta(c_i,c_j),
    \label{eq:modularity}
\end{equation}
where $\omega_{ij}$ is the weight of edge $(i,j)$, the total sum of weight $m=\frac{1}{2}\sum_{(i,j)}\omega_{ij}$, and the sum of weight of the $i$-th node $k_i=\sum_{j}\omega_{ij}$. Intuitively, the term $\frac{k_ik_j}{2m}$ indicates the
probability of edge existing between the nodes $i$ and $j$ in a randomized graph. Typically, when $Q > 0.3$, it indicates a significant community
structure~\cite{b13}.
Greedy modularity maximization~\cite{b13}, recursively joins the pair of communities that increases the modularity most.
Nevertheless, this greedy algorithm often yields super-communities with a significant number of nodes, potentially exceeding the limit of $q$ and resulting in low modularity outcomes. Furthermore, recursively comparing modularity in a greedy manner imposes a substantial computational burden.

% However, this greedy algorithm tends to produce \textcolor{red}{super-communities that contain a large number of nodes (possibly larger than $q$)}\textcolor{blue}{definition?}, outputs results with low modularity and lacks \textcolor{red}{computational efficiency}\textcolor{blue}{it is not apparent }.

Different from the above methods, our method utilizes Louvain algorithm which embraces intrinsic multi-level nature, applicable to large networks with high modularity and is extremely fast~\cite{b14}. To save computation burden, Louvain algorithm defines the gain in modularity $\Delta Q$ by moving an isolated vertex $i$ into a community $c_j$ as,
%\begin{equation}
%    \Delta Q = \left[\frac{k_{i,c_j}}{2m} - \left(\frac{\sum_{tot}+k_i}{2m}\right)^2 \right]+\left[ \left(\frac{\sum_{tot}}{2m}\right)^2+\left(\frac{k_i}{2m}\right)^2\right]\\
%\end{equation}
\begin{align}
    \Delta Q &= \left[\frac{\sum_{p,q\in c_j}\omega_{pq}+\sum_{u\in c_j}\omega_{iu}}{2m}-\left(\frac{\sum_{u\in c_j}k_u+k_i}{2m}\right)^2\right] \nonumber\\
    &\quad - \left[\frac{\sum_{p,q\in c_j}\omega_{pq}}{2m}-\left(\frac{\sum_{u\in c_j}k_u}{2m}\right)^2-\left(\frac{k_i}{2m}\right)^2\right] \nonumber\\
    & = \frac{\sum_{u\in c_j}\omega_{iu}}{2m} - \frac{\sum_{u\in c_j}{k_u} k_i}{2m^2}.
\end{align}
The concomitant loss by moving vertex $i$ out of community $c_k$ is similarly calculated.
In practice, one evaluates the change in modularity by first removing vertex $i$ from its current community and subsequently relocating it to a neighboring community.
Therefore, the quantified impact of moving vertex $i$ from one community to another is attained for decision. The combination that brings the largest $\Delta Q$ is chosen for action until all $\Delta Q$ approaches zero. %Afterwards, \textcolor{magenta}{a compressed network} where each vertex represents a community is generated for further optimization of modularity.
Subsequently, a compressed network is generated in which each vertex represents a community, facilitating further optimization of modularity.
In the compressed network, the weights of the links between the new nodes are determined by the summation of the weights of the links connecting nodes from the respective two communities. If there are links between nodes within the same community, these links result in self-loops within that community in the new network. The two steps are iterated until convergence of modularity or the number of nodes in each community reaches the allowed maximum number of qubits $q$.

\subsection{Community Representation and Update}\label{32}
Community representation compresses each subgraph for further combination. Based upon the results of community detection, within each community, the vertex sets $C_k$ are classified into two categories, namely $C^{in}_k$ and $C^{out}_k$, which satisfies $C^{in}_k \bigcup C^{out}_k = C_k$ and $C^{in}_k \bigcap C^{out}_k = \varnothing$. An in-node set $C^{in}_k$ is defined as the set of vertices whose neighborhoods are also within community $C_k$, while an out-node set $C^{out}_k$ is defined as the set of vertices whose neighborhoods have at least one vertex in the vertex set of another community. Mathematically, define $N(i)$ as the set that contains all the neighbourhoods of vertex $i$, then $C^{in}_k=\{i|N(i)\subseteq C_k\}$, and $C^{out}_k=\{i|\exists j\in N(i),\,j\notin C_k\}=C_k\backslash C^{in}_k$. In the paradigm of distributed QAOA, conventional QAOA is applied to obtain the local solution of subgraph $G_k$ composed of $C_k$ and edges that connect
these vertices, and we obtain a bitstring $z^*_k\in \mathbb{\bar B}^{|C_k|}$, indicating the approximate solution of $G_k$.
%Note that, due to $\mathbb{Z}_2$ symmetry in MaxCut problems~\cite{b15}, $\overline{z^*_k}$ which flips every bit of $z^*_k$, either from $+1$ to $-1$ or $-1$ to $+1$, also denotes the same approximate MaxCut.
% From another perspective, what we are sure about is the logic relationship for vertices within the subgraph, namely whether some two vertices belong to the same category or not.
%\textcolor{magenta}{Existing methods divide an original graph into several subgraphs through random partition or greedy modularity maximization algorithm. Solution of each subgraph is obtained by conventional QAOA, after which each subgraph is viewed as a single vertex and the global solution is obtained by approximately optimally merging local solutions, which is again a problem that can be solved by QAOA. However, note that this brute combination of solutions for each subgraph despise connections (weighted edges) between subgraphs and there are conditions when the weights of inter edges outweigh that of inner edges, resulting in a different yet more favoured classification of vertices within a subgraph.}
% Hence, in our method, we only initially fix the logic relationship within $C^{in}_k$, yet \textcolor{magenta}{the logic relationship} within $C^{out}_k$ and between $C^{in}_k$ and $C^{out}_k$ can be altered.
Each $C^{in}_k$ is represented by a new vertex $I_k$ whose classification $\{\mathbbm{-1},\mathbbm{+1}\}$ denotes $\overline{{z^{in}_k}}$ or ${z^{in}_k}$ of $C^{in}_k$ when combined into the global solution.
Here, $\overline{{z^{in}_k}}$ is defined as the result of flipping every bit of ${z^{in}_k}$, either from $+1$ to $-1$ or from $-1$ to $+1$.
$C^{out}_k$ and $I_k$, for $k=1,2,\ldots,M$, along with corresponding edges, compose a compressed substitute for subgraph $G_k$, where the weight of edge between a vertex in $C^{out}_k$ and $I_k$ is defined as $(-\check{v}_k(\overline{z^{in}_k}\oplus z^{out}_k)+\check{v}_k(z^{in}_k\oplus z^{out}_k))/2$ where $\check{v}_k$ is the local objective function of subgraph $G_k$. The expression means if the local solution is altered due to global connections the amount between before and after is added to $\check{v}_k$.

These compressed subgraphs are merged into a new sequence of subgraphs, where global update utilizes conventional QAOA again as the solver to obtain a higher-level solution, which updates $z^*_k$ into $z^*_{k,u}$ in every subgraph. $z^*_{k,u}$ outperforms $z^*_k$ since a higher-level objective function is taken into account.

Once an updated  $z^*_{k,u}$ is derived from higher-level, local update is performed after which there occur two circumstances. For each subgraph $G_k$, if $z^*_{k,u}=z^*_{k}$ or $z^*_{k,u}=\overline{z^*_{k}}$, then the procedure stops here, and $z^*_{k,u}$ is the local solution when combined into the higher-level solution; otherwise, we fix the updated $z^{out}_{k,u}$ to update $z^{in}_{k}$ into $z^{in}_{k,u}$, and $z^{in}_{k,u}\oplus z^{out}_{k,u}$ is the local solution when combined into the higher-level solution.
%下面的定理ii）就是证明这个
For simplicity of mathematics, in the following, we also denote $z^{*}_{k,u}$ in the first circumstance as $(z^{in}_{k,u}\oplus z^{out}_{k,u})$, where $z^{in}_{k,u}$ and $z^{out}_{k,u}$ are truncated bitstring of $z^*_{k,u}$ as to $C^{in}_k$ and $C^{out}_{k}$. Finally, the approximate solution of the original graph is built up by $\oplus^M _{k=1}(z^{in}_{k,u}\oplus z^{out}_{k,u})$. The following theorem
shows that our method ensures the approximate solution of the original graph outperforms naive combination of approximate local solutions. Specifically, $i)$ indicates the importance of global update, $ii)$ manifests the necessity of local update after global update, and $iii)$ accomplishes the entire proof.
%这里是在介绍定理的三个部分分别在证明什么
\begin{thm} \label{thm:1}Let the original Graph $G$ be partitioned into $M$ subgraphs $\{G_k\}^M_{k=1}$, let $\check{v}$ be the objective function of $G$ w.r.t. Eq.~(\ref{isingvcheck}), and $\check{v}_k$ be the objective function of $\{G_k\}$. To combine local solutions into a global solution, let $s=(s_1,s_2,\ldots,s_M)\in \{-\mathbbm{1},+\mathbbm{1}\}^M$ be the indicator of logic relationship for each vertex set $C_k$ of subgraph $G_k$ or each in-node set $C^{in}_k$ of subgraph $G_k$. Let $b_k$ be a bitstring describing a kind of logic relationship, if $s_k = -\mathbbm{1}$, $s_k(b_{k})=\overline{b_{k}}$; if $s_k = +\mathbbm{1}$, $s_k(b_{k})=b_{k}$. Three formulae hold,\\
$i)$ $\check{v}(\oplus^M_{k=1} z^*_{k,u})\leq \min_s\check{v}(\oplus^M_{k=1} s_k(z^*_{k}))$, \\
$ii)$ $\check{v}_k(z^{in}_{k,u}\oplus z^{out}_{k,u})\leq \min_{s_k}\check{v}_k(s_k(z^{in}_{k})\oplus z^{out}_{k,u})$, \\
%iii) $f(\oplus^M _{k=1}(z^{in}_{k,u}\oplus z^{out}_{k,u}))\leq \min_sf(\oplus^M _{k=1}(s_k(z^{in}_{k})\oplus z^{out}_{k,u}))=f(\oplus^M_{k=1} z^*_{k,u})\leq min_sf(\oplus^M_{k=1} s_k(z^*_{k}))$.
$iii)$ $\check{v}(\oplus^M _{k=1}(z^{in}_{k,u}\oplus z^{out}_{k,u}))\leq \min_s\check{v}(\oplus^M_{k=1} s_k(z^*_{k}))$.
\end{thm}

\begin{proof}
i) Denote $\mathcal{D}$ as the feasible domain of solutions. Since each $z^*_k \in \mathcal{D}(z^*_k)$ can be truncated into the bitstring of in-node set $z^{in}_k$ and the bitstring of out-node set $z^{out}_k$, namely $z^*_k = z^{in}_k \oplus z^{out}_k$, each $z^*_k$ is within $\mathcal{D}(z^*_{k,u})$, namely $\mathcal{D}(z^*_k)\subseteq\mathcal{D}(z^*_{k,u})$, so is $\bigcup^M_{k=1}\mathcal{D}(z^*_k)\subseteq\bigcup^M_{k=1}\mathcal{D}(z^*_{k,u})$. Since QAOA approximates the optimal solution within a domain, $\oplus^M_{k=1} z^*_{k,u}$ is at least as good as $\oplus^M_{k=1} s_k(z^*_{k})$, which leads to $\check{v}(\oplus^M_{k=1} z^*_{k,u})\leq \min_s\check{v}(\oplus^M_{k=1} s_k(z^*_{k}))$. \\
ii) By definition, $z^{in}_{k,u}$ is updated from $z^{in}_{k}$ after $z^{out}_{k,u}$ is updated from $z^{out}_{k}$. When we fix $z^{out}_{k,u}$ in subgraph ${G_k}$, $z^{in}_{k}$ is updated into $z^{in}_{k}$ only when $z^{in}_{k,u}$ outperforms $z^{in}_{k}$ concerning minimization the local objective function $\check{v}_k$, which leads to $\check{v}_k(z^{in}_{k,u}\oplus z^{out}_{k,u})\leq \min_{s_k}\check{v}_k(s_k(z^{in}_{k})\oplus z^{out}_{k,u})$.\\
iii) Denote $\check{v}_{inter}$ as the objective function of edges between subgraphs. Aggregate ii) from $k=1$ to $M$, then add all the interconnection edges between different subgraphs, we get $\sum^M_{k=1}\check{v}_k(z^{in}_{k,u}\oplus z^{out}_{k,u})+\check{v}_{inter}(\oplus^M_{k=1}z^{out}_{k,u})\leq \sum^M_{k=1}\min_{s_k}\check{v}_k(s_k(z^{in}_{k})\oplus z^{out}_{k,u})+\check{v}_{inter}(\oplus^M_{k=1}z^{out}_{k,u})$. Note that the objective function $\check{v}$ can be divided into terms within the subgraph and between subgraphs. Represented by $\check{v}$, we have $\sum^M_{k=1}\check{v}_k(z^{in}_{k,u}\oplus z^{out}_{k,u})+\check{v}_{inter}(\oplus^M_{k=1}z^{out}_{k,u})=\check{v}(\oplus^M _{k=1}(z^{in}_{k,u}\oplus z^{out}_{k,u}))$ and $\sum^M_{k=1}\min_{s_k}\check{v}_k(s_k(z^{in}_{k})\oplus z^{out}_{k,u})+\check{v}_{inter}(\oplus^M_{k=1}z^{out}_{k,u})=\min_s\check{v}(\oplus^M_{k=1}s_k(z^{in}_{k})\oplus z^{out}_{k,u})$, which leads to $\check{v}(\oplus^M _{k=1}(z^{in}_{k,u}\oplus z^{out}_{k,u}))\leq \min_s\check{v}(\oplus^M _{k=1}(s_k(z^{in}_{k})\oplus z^{out}_{k,u}))$. Since the global update does not change the logic relationship within $C^{in}_k$, along with i), $\min_s\check{v}(\oplus^M _{k=1}(s_k(z^{in}_{k})\oplus z^{out}_{k,u}))=\check{v}(\oplus^M_{k=1} z^*_{k,u})\leq \min_s\check{v}(\oplus^M_{k=1} s_k(z^*_{k}))$.
\end{proof}
\begin{example}
For the sake of conciseness, we use a 9-node Max-Cut problem as an example. The objective function is expressed as,
\begin{align}
    \min \check{v}(n_0,\cdots,n_8)=&\;n_0n_1+n_1n_2+n_0n_2+n_2n_3+n_3n_4\nonumber\\
    &+n_2n_4+n_5n_6+n_6n_8+n_7n_8+n_5n_7\nonumber\\
    &+2n_3n_5+2n_4n_6,
\end{align}
where $n_i \in \mathbb{\bar B}, i \in \{0,\cdots,8\}$.
Suppose quantum devices allow a maximum number of $6$ qubits. Also, the original graph $G$ is partitioned into $C_1$ and $C_2$, shown in Fig.~\ref{fig:example}. As to the existing method, in ascending order of numbers in nodes, local solutions are $z^*_1=\{(-1)(-1)(+1)(-1)(-1)\}$ and $z^*_2=\{(-1)(+1)(+1)(-1)\}$, and the indicator $s = (+\mathbbm{1},+\mathbbm{1})$. The global solution $s_1(z^*_1) \oplus s_2(z^*_2)$ yields the value of objective function $\check{v} = \check{v}_1 + \check{v}_2 + \check{v}_{inter} = -2 -4 + 0 = -6$. However, this is not the optimal case. As to our method, vertex set $C^{in}_1=\{0,1,2\}$ and $C^{out}_1=\{3,4\}$, $C^{in}_2=\{7,8\}$ and $C^{out}_2=\{5,6\}$. $C^{in}_1$ is suppressed as a single vertex, so is $C^{in}_2$. The combination of suppressed $C_1$ and $C_2$ is again solved by QAOA and generates $z^*_{1,u}$ and $z^*_{2,u}$ to update local solutions.
On one hand, $z^*_{1,u}=\{(-1)(-1)(+1)(-1)(+1)\}$ which neither equals to $z^*_1$ nor  $\overline{z^*_1}$, so $z^{in}_1=\{(-1)(-1)(+1)\}$ should be updated into $z^{in}_{1,u}=\{(+1)(-1)(+1)\}$. On the other hand, $z^*_{2,u}=\{(+1)(-1)(-1)(+1)\}$ which equals to $\overline{z^*_1}$, so there is no need to update $z^{in}_2$. The global solution yields $\check{v} = \check{v}_1 + \check{v}_2 + \check{v}_{inter} = -2 -4 - 4 = -10 < -6$, which is also the optimal solution. Fig.~\ref{fig.examplecom} demonstrates the solutions of the existing method and our method.
\end{example}
\begin{figure}[htbp]
    \centering
    \includegraphics[width=0.35\textwidth]{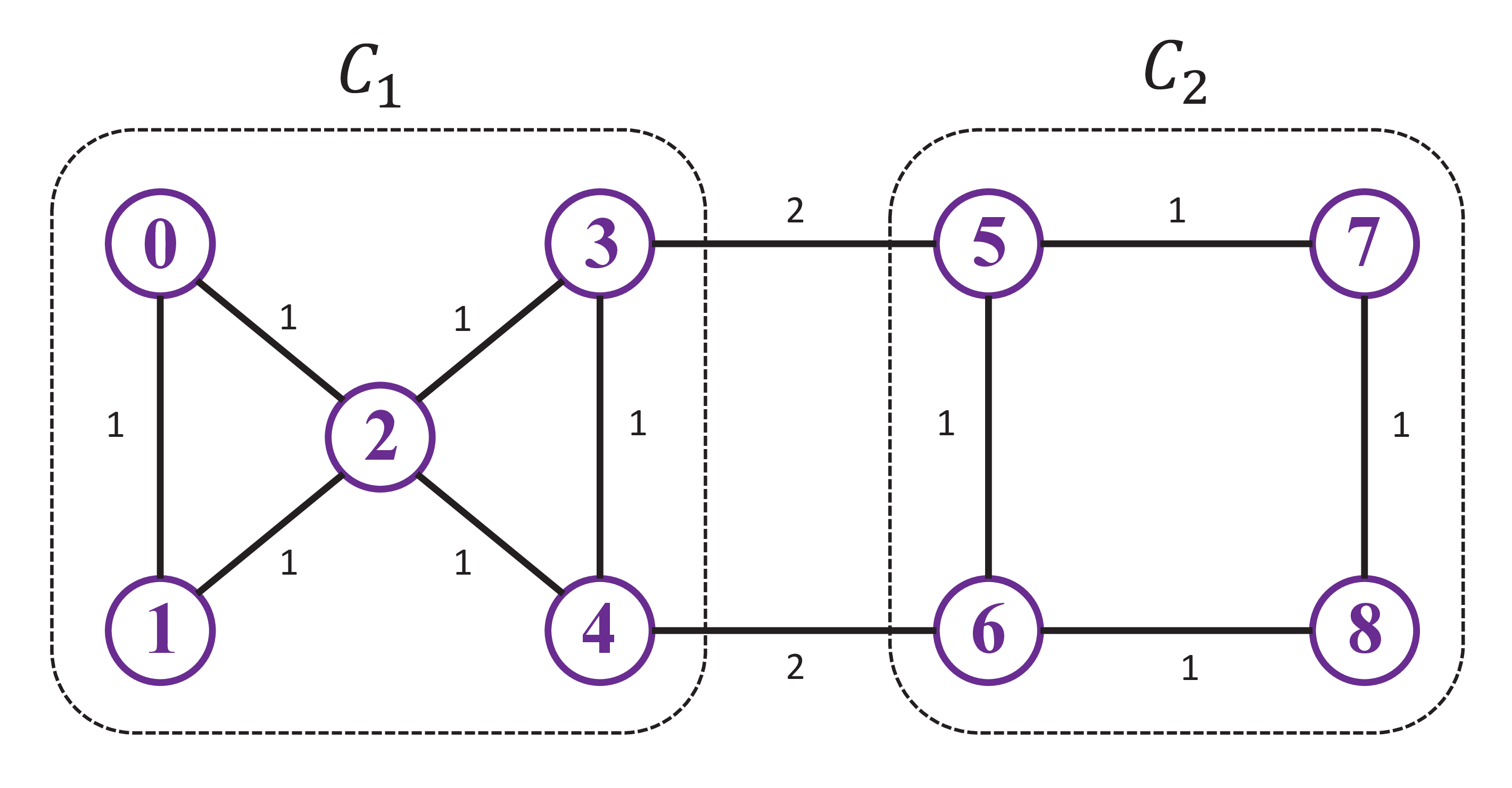}
    \caption{An example of Theorem~\ref{thm:1}. The graph is an undirected graph and the number beside the edge denotes weight.}
    \label{fig:example}
\end{figure}

\begin{figure}[htbp]
    \centering
    \subfigure[Cut derived by the existing method.]{
        \includegraphics[width=0.35\textwidth]{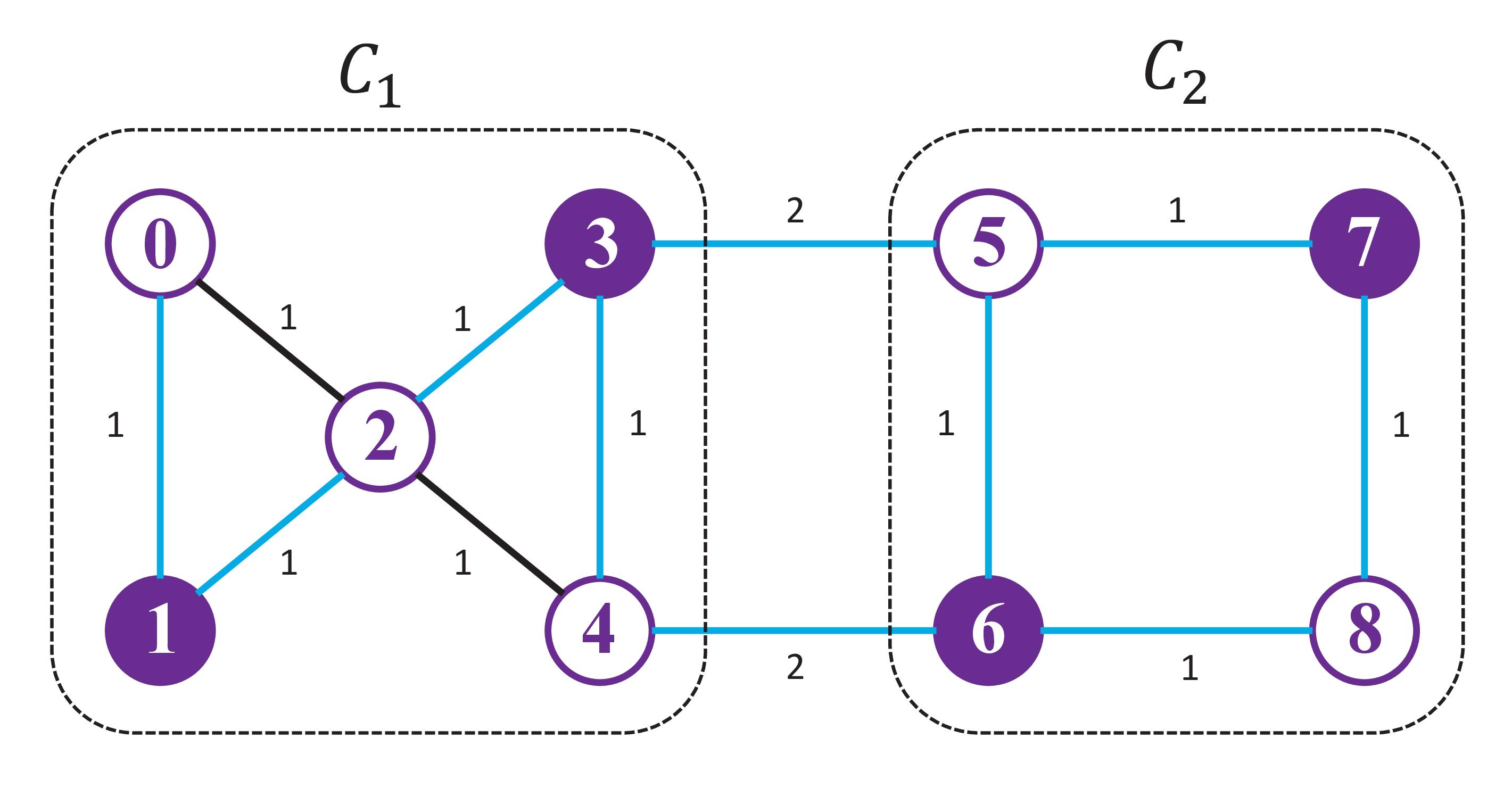}
    }
    \quad    %用 \quad 来换行
    \subfigure[Cut derived by our method.]{
    	\includegraphics[width=0.35\textwidth]{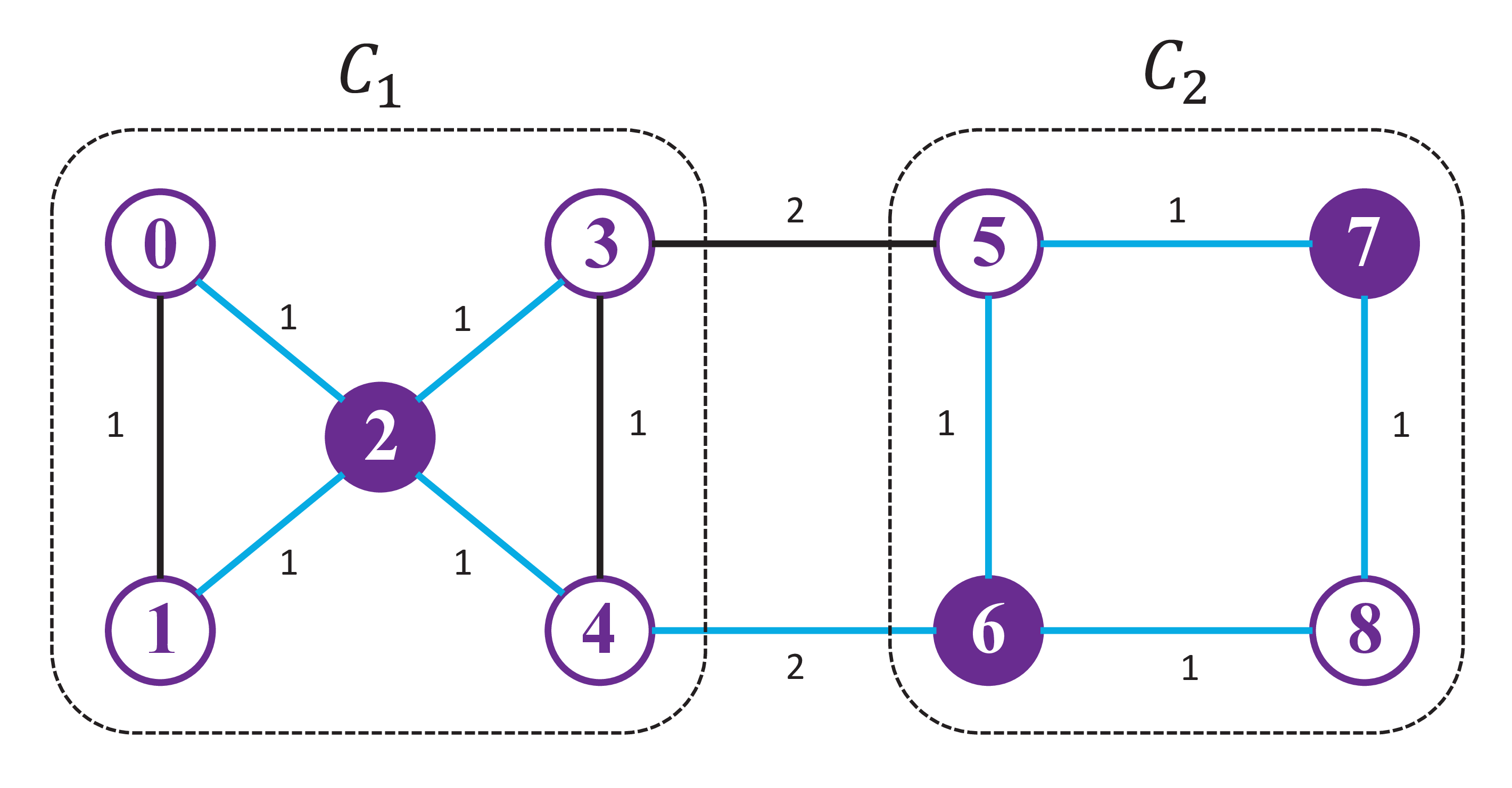}
    }
    \caption{Comparison of cut between the existing method and our method. The blue line denotes a cutting edge, the two vertices of which belong to different subsets of vertices.}
    \label{fig.examplecom}
\end{figure}

\begin{rem}
Different partitions of the original graph would lead to varied final values of $\check{v}$. For generality of comparison between the existing method and our method, we conduct $20000$ parallel experiments for both partition techniques. As for random partition, $4.9\%$ yields $-4$, $8.5\%$ yields $-6$, $41.9\%$ yields $-8$, and $44.7\%$ yields $-10$. Since each final value corresponds to several partitions, we offer one feasible partition for each value here for better understanding. Possible partitions are %$\{(-1)(+1)(-1)(+1)(-1)(-1)(-1)(+1)(+1)\}, \{(-1)(+1)\\(+1)(-1)(+1)(+1)(-1)(+1)(+1)\}, \{(-1)(+1)(+1)(+1)\\(-1)(-1)(+1)(+1)(+1)\}, \{(-1)(-1)(+1)(+1)(-1)(-1)\\(+1)(+1)(-1)\}$,
$\{(-1)(+1)(-1)(+1)(-1)(-1)(-1)(+1)(+1)\}$ for $-4$, $\{(-1)(+1)(+1)(-1)(+1)(+1)(-1)(+1)(+1)\}$ for $-6$, $\{(-1)(+1)(+1)(+1)(-1)(-1)(+1)(+1)(+1)\}$ for $-8$, and $\{(-1)(-1)(+1)(+1)(-1)(-1)(+1)(+1)(-1)\}$ for $-10$. As for greedy modularity maximization, $31.9\%$ yields $-6$, and $68.1\%$ yields $-10$. Possible partitions are $\{(-1)(-1)(+1)(-1)(-1)(+1)(-1)(-1)(+1)\}$ for $-6$ and $\{(+1)(-1)(-1)(+1)(-1)(-1)(+1)(+1)(-1)\}$ for $-10$. In contrast, our method yields $-10$ for $100\%$. Possible partition is $\{(+1)(-1)(+1)(-1)(+1)(+1)(-1)(-1)(+1)\}$.
\end{rem}

%紧跟着例子应该更好

\subsection{Workflow of Distributed QAOA}\label{33}
%In this part, we summarize the workflow of distributed QAOA from a holistic perspective, which includes how to combine local solutions and update through layers of combination procedure.

Since existing methods directly merge local solutions into a global solution, they do not take into account the limited number of qubits, which is likely to be overtaken when the number of subgraphs is large. Hence, we propose a tree-structured bottom-up combination procedure. The entire workflow of distributed QAOA is presented in Algorithm~\ref{alg:qaoa}, i) the original graph $G$ is partitioned into $\{G^{(1)}_k\}_{k=1}^{M_1}$, where the superscript denotes the height of tree-structured combination process; ii) QAOA is applied to each subgraph $\{G^{(i)}_k\}$, which is further compressed as $\{SG^{(i)}_k\}$ by representing the in-node set as a new vertex; iii) each $\{SG^{(i)}_k\}$ then joins its adjacent compressed graph as long as the required number of qubits for the joined community is less than the allowed number of qubits $q$; iv) ii) and iii) are iterated until the number of communities equals to one. The update process is conducted from top to bottom. From the second layer to the bottom layer whose height equals to one, each ${z^{(i)}_k}^*$ is updated into ${z^{(i)}_{k,u}}^*$, after which the approximate global solution $z_G^*=\oplus^{M_1}_{k=1}{(z^{(1)}_{k,u}}^{in} \oplus {z^{(1)}_{k,u}}^{out})$.
\IncMargin{0.6em}
\begin{algorithm}[ht]
\SetKwInOut{Input}{Input}\SetKwInOut{Output}{Output}
\caption{Distributed QAOA for pseudo-Boolean problems}\label{alg:qaoa}
\Input{a multilinear polynomial $f$ of a pseudo-\\Boolean problem, maximum qubit count $q$}
\Output{approximate optimal solution $z_f^*$ of $f$}
\BlankLine
Eliminate uncoupled variables of $f\rightarrow\hat{f}$\;
Quadratize $\hat{f}$ and eliminate chains$ \; \rightarrow \check{v}$\;
Formulate $\check{v}$ as an Ising model $\rightarrow$ graph $G$\;

$G$ $\rightarrow$ $\{G^{(1)}_k\}_{k=1}^{M_1}$, $|C^{(1)}_k|\leq q$\\
$i=1$\\
\While(\tcp*[f]{combine local solutions}){$M_i\neq
1$}{\For{$k$ = $1,2,\ldots,M_i$}{
Apply QAOA to $\{G^{(i)}_k\}\rightarrow {z^{(i)}_k}^*$\;
Compress each $\{G^{(i)}_k\}$ as $\{SG^{(i)}_k$\}\;}
$\{SG^{(i)}_k\}^{M_1}_{k=1}\xrightarrow{join}\{G^{(i+1)}_k\}_{k=1}^{M_{i+1}}, |C^{(i)}_k| \leq q,M_i<M_{i+1}$\;
$i \leftarrow i+1$\;}
\While(\tcp*[f]{update local solutions}){$i \neq 1$}{
$i \leftarrow i-1$\;
\For{$k$ = $1,2,\ldots,M_i$}{
Derive ${z^{(i)}_{k,u}}^*$ from ${z^{(i+1)}_{k}}^*$\;
\eIf{${z^{(i)}_{k,u}}^*=={z^{(i)}_k}^*$ or ${z^{(i)}_{k,u}}^*==\overline{{z^{(i)}_k}^*}$}{Truncate ${z^{(i)}_{k,u}}^*$ into ${z^{(i)}_{k,u}}^{in} \oplus {z^{(i)}_{k,u}}^{out}$\;} {
Update ${z^{(i)}_{k,u}}^*$ into ${z^{(i)}_{k,u}}^{in} \oplus {z^{(i)}_{k,u}}^{out}$\;}
}}

$z_G^*=\oplus^{M_1}_{k=1}{(z^{(1)}_{k,u}}^{in} \oplus {z^{(1)}_{k,u}}^{out})$\;
Assign values to preprocessed variables, $z_G^* \rightarrow z_f^*$\;
\end{algorithm}
Note that our method is not only applicable to solving large Max-Cut instances, but is also capable of solving large pseudo-Boolean optimization problems. As mentioned above, an optimization problem is written as a pseudo-Boolean function, which is converted into a multilinear polynomial. Subsequently, through elimination of uncoupled variables, quadratization and elimination of chains, the initial problem is
reduced into a simplified Ising model. After the global problem is solved, these preprocessed variables are assigned with values to attain the solution to the initial optimization problem.
\iffalse
\IncMargin{0.25em}
\begin{algorithm}
\SetKwInOut{Input}{Input}\SetKwInOut{Output}{Output}
\caption{Distributed QAOA}\label{alg:qaoa}
\Input{original graph $G$, allowed number of qubits $q$}
\Output{graph $G$ Max-Cut's partition $z^*$}
\BlankLine

$G$ $\rightarrow$ $\{G^{(1)}_k\}_{k=1}^{M_1}$, $|C^{(1)}_k|\leq q$\\
$i=1$\\
\While(\tcp*[f]{combine local MaxCuts}){$M_i\neq
1$}{\For{$k$ = $1,2,\ldots,M_i$}{
Apply QAOA to $\{G^{(i)}_k\}\rightarrow {z^{(i)}_k}^*$\;
Compress each $\{G^{(i)}_k\}$ as $\{SG^{(i)}_k$\}\;}
$\{SG^{(i)}_k\}^{M_1}_{k=1}\xrightarrow{join}\{G^{(i+1)}_k\}_{k=1}^{M_{i+1}}, |C^{(i)}_k| \leq q,M_i<M_{i+1}$\;
$i \leftarrow i+1$\;}
\While(\tcp*[f]{update local MaxCuts}){$i \neq 1$}{
$i \leftarrow i-1$\;
\For{$k$ = $1,2,\ldots,M_i$}{
\eIf{${z^{(i)}_{k,u}}^*=={z^{(i)}_k}^*$ or ${z^{(i)}_{k,u}}^*==\overline{{z^{(i)}_k}^*}$}{Truncate ${z^{(i)}_{k,u}}^*$ into ${z^{(i)}_{k,u}}^{in} \oplus {z^{(i)}_{k,u}}^{out}$\;} {
Update ${z^{(i)}_{k,u}}^*$ into ${z^{(i)}_{k,u}}^{in} \oplus {z^{(i)}_{k,u}}^{out}$\;}
}}

$z^*=\oplus^{M_1}_{k=1}{(z^{(1)}_{k,u}}^{in} \oplus {z^{(1)}_{k,u}}^{out})$\;
\end{algorithm}
\fi

%流程图移到了前面，算法还是放到这里，否则开头要介绍很多符号

\section{Numerical Results}\label{4}
In this section, we conduct numerical simulation to verify the effectiveness and generality of distributed QAOA. Concretely, since the existing methods only focus on Max-Cut problems, we first focus on large-scale Max-Cut problems for fair comparison. Afterwards, we give an instance of our method solving a general pseudo-Boolean problem.
\subsection{Comparison on Max-Cut}
As for experimental settings, four kinds of graphs with different degrees of nodes and weights of edges are generated to make comprehensive comparison between distributed QAOA and the existing method. Specifically, we fix the node size to be $100$; an unweighted $9$-regular graph (whose degree of each node is $9$) is denoted as UR, while a weighted $9$-regular graph with weight of each edge uniformly sampled from $[1,6]$ is denoted as WR; an unweighted Erdos-Renyi (ER) graph (whose degree of node is randomly assigned) with an average degree of each node assigned as $5$ is denoted as UE, while a weighted ER graph with an average degree of each node assigned as $5$ and weight of each edge uniformly sampled from $[1,6]$ is denoted as WE. The maximum qubit count regarding a local quantum device, namely $q$, is configured as $10$, and $1$-layer QAOA with $20$ iterations are used for training local Max-Cut solvers.
We use the approximation ratio
\begin{equation}
    r = \frac{V_{max}-\langle \psi(\vec{\gamma}^*,\vec{\beta}^*)|H|\psi(\vec{\gamma}^*,\vec{\beta}^*) \rangle}{V_{max}-V_{min}}\in [0,1],
\end{equation}
as a measure of how close the final state is to the optimal solution, where $V_{max}$ and $V_{min}$ are the theoretical maximum and minimum values of the original objective function. Clearly, a larger $r$ indicates a better solution decoded by the final ansatz state $|\psi(\vec{\gamma}^*,\vec{\beta}^*)\rangle$.
%我看别人的文章此处
$V_{min} = 0$ since there is no cut edge if all vertices are classified into the same group. $V_{max}$ is acquired by the cut value searched by Goemans-Williamson (GW) algorithm~\cite{b16} which is the most renowned classical Max-Cut solver that utilizes semi-definite programming (SDP)~\cite{b17}. Since the existing method already manifests its advantages over GW algorithm, we focus on the advantages of our method over the existing method in the simulation.

Compared with the existing method, our method makes advances in two main aspects, namely community detection and accommodation (CDA for short), as well as community representation and update (CRU for short). To validate the effectiveness of both aspects, we have done multiple ablation studies. In the following two figures, each box is plotted with $20$ data points from $20$ parallel experiments, where $1.5$ interquartile range (IQR), median and mean values (the colored numbers) are shown.

In Fig.~\ref{fig:sim1}(a), random partition is applied to the four graphs, and only CRU and no CDA is utilized in our method. In Fig.~\ref{fig:sim1}(b), greedy modularity maximization is applied to the four graphs, and only CRU and no CDA is utilized in our method. Our method outperforms the existing method in all instances, whether median or mean values are considered, which validates the effectiveness of CRU and verifies Theorem~\ref{thm:1} from the perspective of experiments. Specifically, in weighted regular graphs or weighted ER graphs, our method achieves significant superiority against the existing method, compared with unweighted counterparts. This is because CRU takes into account the connectivity between communities and updates local Max-Cut solvers with global heuristics. For both subfigures, regular graphs have lower approximation ratio than ER graphs in that regular graphs tend to have more connection between communities which degrades local Max-Cut solvers. Also, weighted graphs have lower approximation ratio than ER graphs due to weighted inter-community edges and more complicated Max-Cut problems for local Max-Cut solvers.

\begin{figure}[ht]
    \centering
    \includegraphics[width=0.48\textwidth]{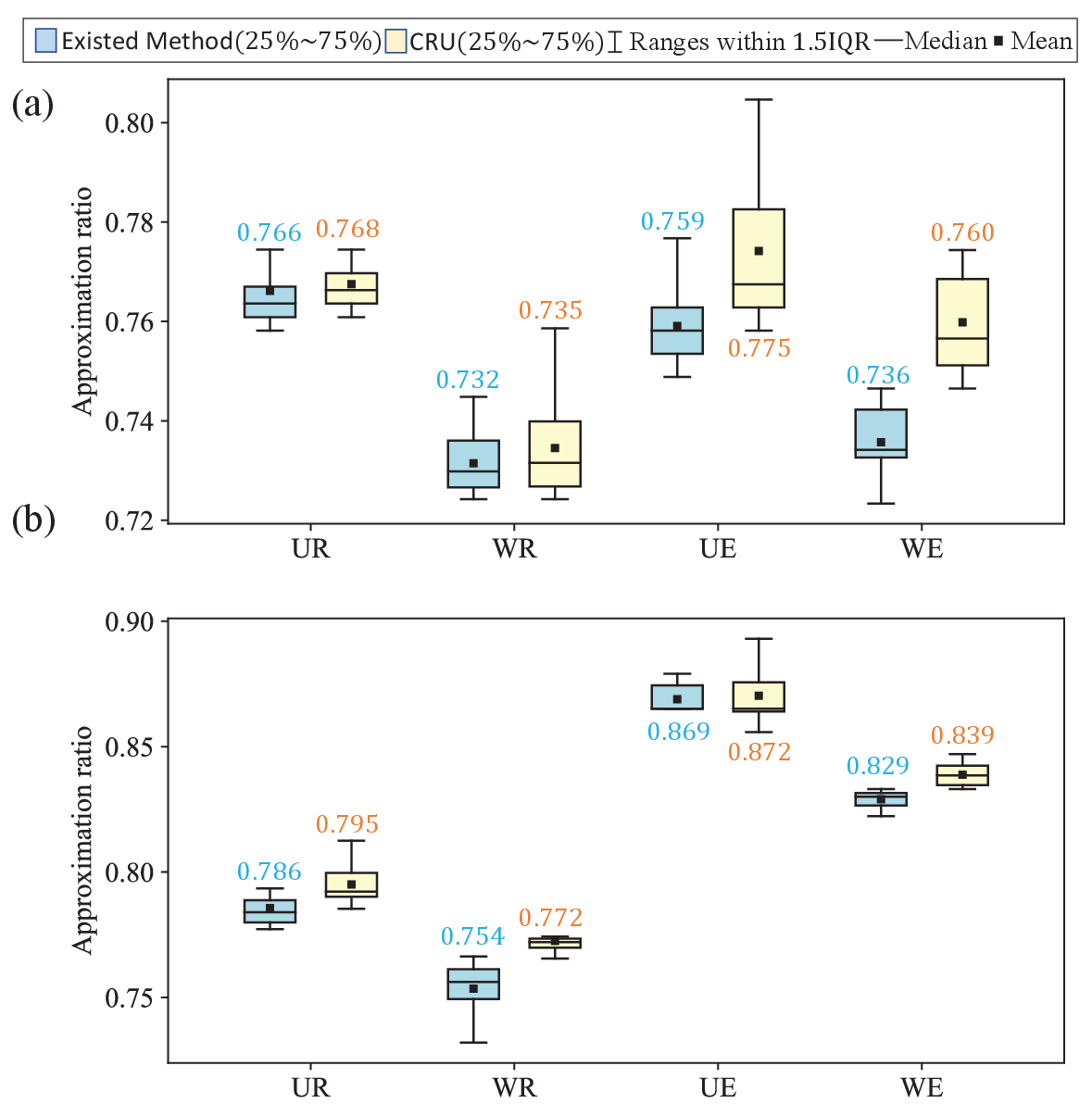}
    \caption{Comparison between the existing method and our method with only CRU. (a) uses random partition technique; (b) employs greedy modularity maximization technique.}
    \label{fig:sim1}
\end{figure}
In Fig.~\ref{fig:sim2}(a), Louvain algorithm is applied to the four graphs, so both CRU and CDA are utilized in out method. While similar results are demonstrated, the approximation ratio is higher than that in Fig.~\ref{fig:sim1}, which manifests the effectiveness of CDA. In Fig.~\ref{fig:sim2}(b), all three graph partition techniques, along with CRU are applied to the four graphs. Median values of our method in Fig.~\ref{fig:sim1}(a), (b) and Fig.~\ref{fig:sim2}(a), along with the modularity for the partitioned community are presented. This clarifies the importance of CDA that a higher level of modularity benefits the overall approximation ratio.
\begin{figure}[ht]
    \centering
    \hspace{-6.7mm}
    \includegraphics[width=0.50\textwidth]{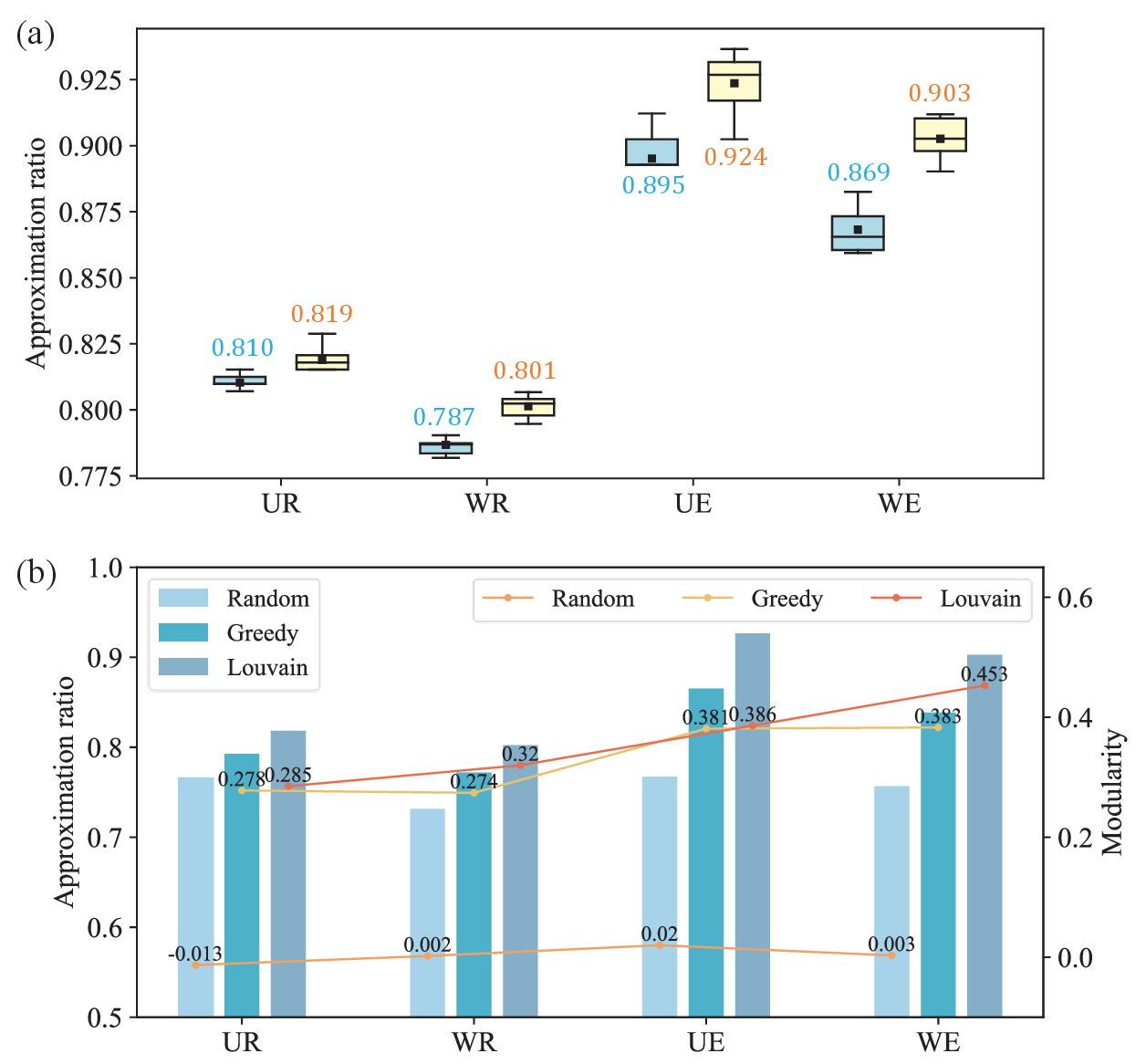}
    \caption{Comparison between the existing method and our method with only CRU, as well as comparison between three partition techniques. (a) uses Louvain algorithm technique (CDA); (b) uses CRU for three partition techniques. The legend is the same as Fig.~\ref{fig:sim1}.}
    \label{fig:sim2}
\end{figure}

\subsection{An Instance of Solving General pseudo-Boolean Problems}
The cubic knapsack problem (CKP) is a variant of the classical quadratic knapsack problem, and is a typical pseudo-Boolean problem. The goal is to select a subset of these items to maximize the total value that contains cubic terms while satisfying certain constraints.
Many problems can be modeled as CKP, including the satisfiability problem Max 3-SAT~\cite{943}, the selection problem~\cite{944}, and the network alignment problem~\cite{945}.
Here, we give an instance of solving a CKP, where we reduce the problem from a general pseudo-Boolean problem into a simplified Ising model, utilize distributed QAOA to solve the problem and reconstruct the solution to the original CKP.
\begin{equation}
\begin{array}{ll}
    \mathop{\max}
    &2x_1+5x_2+2x_3+2x_4-2x_7+8x_1x_2+6x_1x_3\\&+10x_1x_4+4x_1x_5+2x_2x_3+6x_2x_4+3x_2x_6\\&+4x_3x_4+4x_4x_7+7x_1x_3x_4\\
    \text{s.t.} &8x_1+6x_2+5x_3+3x_4\leq 16
\end{array}
\end{equation}
where $x_i \in \mathbb{B}, i \in [7]$ indicates whether to choose the $i$-th item or not.
Follow the procedure of reducing a general pseudo-Boolean problem into a simplified Ising model. First, the constrained problem is transformed into an unconstrained problem by adding a quadratic penalty to the negative counterpart of the objective function, which is $P_1(8x_1+6x_2+5x_3+3x_2+x_8+2x_9-16)^2$, where $P_1$ is a large positive number and $x_{8,9}\in\mathbb{B}$ are slack variables to convert an inequality constraint into a equality constraint. Second, $x_5$ and $x_6$ are uncoupled variables, which can both be assigned $1$ in the first place. The non-quadratic term $7x_1x_3x_4$ is quadratized by substituting itself with $7[x_4x_{10}-P_2(x_1x_3-2x_1x_{10}-2x_3x_{10}+3x_{10})]$ in the objective function to ensure $x_{10}=x_1x_3$, where $P_2$ is also a large positive number. Next, replace $x_\cdot$ with $\frac{1-z_\cdot}{2}$ to obtain the Ising formulation of the objective function. Since $z_7$ exists only in one term $z_4z_7$, which means $z_7$ can be assigned value depending on $z_4$, namely a delayed decision of $z_7$ satisfying $z_7 = z_4$. Finally, after we choose $P_1=P_2=10$ and neglect the constants and transform the objective function from $\max$ to $\min$, we derive the original graph for the problem, and the Hamiltonian (or Ising model) of the corresponding problem is
\begin{equation}
\check{v}(\mathbf{\hat{z}})=-\mathbf{\hat{z}^T\,M\,\hat{z}},
\end{equation}
where
\begin{gather*}
\mathbf{\hat{z}} = (z_1,z_2,z_3,z_4,z_8,z_9,z_{10})^T,\\
\mathbf{M} =
\begin{bmatrix}
          613 &238 &216   &117.5 &40   &80 &-35\\
            0  &436 &149.5  &88.5  &30   &60 &0\\
            0  &0    &393   &74    &25   &50 &-35\\
            0  &0    &0      &225.5 &15   &30  &-1.75\\
            0  &0    &0      &0      &70   &10  &0\\
            0  &0    &0      &0      &0     &140 &0\\
            0  &0    &0      &0      &0     &0     &-66.5\\
\end{bmatrix}.
\end{gather*}
We apply our method to $\check{v}$, where the original graph is divided into two subgraphs and is solved, respectively, by QAOA with $p=1$. The result is $\{z_1,\ldots, z_{10}\}^*=\{(-1)(+1)(-1)(-1)(-1)(-1)(-1)(+1)(+1)(+1)\}$. By reconstructing the solution to the original CKP problem, we derive the final value of the original objective function as $39$, with the optimal solution $(x_1\ldots x_7)^*=1011111$.
% \begin{gather*}
%     U=\left[ \begin{matrix}
%         \begin{matrix}
%           613 &238 &216   &117.5 &40   &80 &-35\\
%             0  &436 &149.5  &88.5  &30   &60 &0\\
%             0  &0    &393   &74    &25   &50 &-35\\
%             0  &0    &0      &225.5 &15   &30  &-1.75\\
%             0  &0    &0      &0      &70   &10  &0\\
%             0  &0    &0      &0      &0     &140 &0\\
%             0  &0    &0      &0      &0     &0     &-66.25\\
%         \end{matrix}
%       \end{matrix} \right]
% \end{gather*}

\section{Conclusion and Future Work}\label{5}
In this paper, we propose a distributed QAOA targeting at solving pseudo-Boolean problems. We have made two major contributions. First, we present a pipeline to compactly reduce a general pseudo-Boolean problem into a simplified Ising model via elimination of uncoupled variables, quadratization and elimination of chains, which benefits distributed QAOA. Second, we put forward distributed QAOA which refines community detection of existing methods, and innovatively propose community representation and update which brings global heuristics to local QAOA solvers. We prove that our method achieves a larger cut than the existing method, under the same partition of graph, in Theorem~\ref{thm:1}. Also, in our method, the qubit limitation is taken into account
when merging local solutions into an overall solution. Simulation results shows our method outweighs existing methods in multiple kinds of graphs. The ablation study validates the effectiveness of community detection and accommodation, and community representation and update. Our work paves the way for application of distributed QAOA on solving large-scale pseudo-Boolean problems on qubit-limited local quantum devices.

Future studies can further explore this issue by several aspects. First, running time of distributed local QAOA can be optimized. Since the running time of QAOA relies on the number of qubits actually used, local quantum devices could use relatively less qubits to save time. However, the number of quantum devices applied will increase. This requires an overall efficient allocation of partitioned graphs, which strikes a right balance  between running time and the total number of employed quantum devices. %\textcolor{red}{Second, in some particular graphs, the technique of community representation and update could fail to efficiently reduce the original graph to small graphs due to large number of out-nodes in a community. Potential solutions include partially sampling the entire out-nodes and running community detection again.} \textcolor{blue}{this would become a drawback of your work.}
Second, QUBO model of other complicated problems can be devised to enable the problem to be solved by distributed QAOA.
Finally, to actually deploy such a distributed QAOA system, networked control system for distributed quantum devices are necessary. Proposals
addressing such control systems can be found in~\cite{941,942}.
%for distributed quantum algorithms~\cite{b18}, such as distributed QAOA, distributed variational quantum algorithms~\cite{b19,b20} and distributed quantum machine learning~\cite{b21,b22,b23,b24}, since

\vspace{12pt}

\end{document}